\documentclass[3p,11pt]{elsarticle}

\usepackage{lipsum}
\makeatletter
\def\ps@pprintTitle{%
 \let\@oddhead\@empty
 \let\@evenhead\@empty
 \def\@oddfoot{}%
 \let\@evenfoot\@oddfoot}
\makeatother

\usepackage{latexsym}
\usepackage{graphicx}  

\usepackage{subfigure}
\usepackage[T1]{fontenc}
\usepackage[latin1]{inputenc}
\usepackage{dsfont}
\usepackage{amssymb}
\usepackage{amsmath}
\usepackage{amsfonts}
\usepackage{latexsym}
\usepackage{color}

\usepackage{nomencl}
\makenomenclature

\newtheorem{theorem}{Theorem}[section]

\newtheorem{proposition}[theorem]{Proposition}
\newtheorem{corollary}[theorem]{Corollary}

\newtheorem{example}[theorem]{Example}
\newproof{proof}{Proof}
\newproof{pot}{Proof of Theorem \ref{thm2}}

\begin{document}

\begin{frontmatter}

\title{{\bf Reset control systems: the zero-crossing resetting law} \tnoteref{dpi}}

\author[a1]{Alfonso~Ba\~nos\corref{cor1}} 
\ead{abanos@um.es}

\author[a2]{Antonio Barreiro} 
\ead{abarreiro@uvigo.es}

\tnotetext[dpi]{This work has been supported by {\em FEDER-EU} and {\em Ministerio de Ciencia e Innovaci\'on (Gobierno de Espa\~na)} under project DPI2016-79278-C2-1-R, and {\em Fundaci\'on S\'eneca (Comunidad Aut\'onoma de la Regi\'on de Murcia)} under project 20842/PI/18.
}

\cortext[cor1]{Corresponding author}

\address[a1]{Universidad de Murcia, Dept. Inform\'atica y Sistemas, 30100 Murcia, Spain}
\address[a2]{Universidad de Vigo, Dept. Ingenier\'ia de Sistemas y Autom\'atica, 36310 Vigo, Spain}

%%%%%%%%%%%%%%%%%%%%%%%%%%%%%%%%%%%%%%%%%%%%%%%%%%%%%%%%%%%%%%%%%%%%%%%%%%%%%%%%
\begin{abstract} A novel representation of reset control systems with a zero-crossing resetting law, in the framework of hybrid inclusions, is postulated. The problems of well-posedness and stability  of the resulting hybrid dynamical system are investigated, with a strong motivation in how non-deterministic behavior is accomplished in control practice. Several stability conditions, based on the eigenstructure of matrices related with periods of the reset interval sequences, and on Lyapunov function-based conditions, are developed.  
\
\end{abstract}
%%%%%%%%%%%%%%%%%%%%%%%%%%%%%%%%%%%%%%%%%%%%%%%%%%%%%%%%%%%%%%%%%%%%%%%%%%%%%%%%


\end{frontmatter}

{\em Notation}: 
$\mathds{R}_{\geq 0}$ is the set of non-negative real numbers, $\mathds{R}^n$ is the $n$-dimensional Euclidean space, and $\mathbf{x} = (x_1,\cdots,x_n) \in \mathds{R}^n$ is a column vector; $\|\mathbf{x}\|$ is the euclidean norm.  For a matrix $A \in \mathds{R}^n\times \mathds{R}^m$, $\|A\|$ is its spectral norm.
 $\mathds{B}$ is the closed unit ball in $\mathds{R}^n$ centered at the origin. $S^n = \{\mathbf{x} \in \mathds{R}^{n+1}: \|\mathbf{x}\|=1 \}$ is the unit $n$-sphere. The set of $n \times n$ (positive definite) symmetric matrices is denoted by ${\mathds S}^n$.  % For a set $K \subset \mathds{R}^n$, $\overline{K}$ is its closure, and $\text{int }K$ is the interior of $K$. 
 $\mathcal{S}_{\mathcal H}(\xi)$ is the set of maximal solutions $\phi$ to the hybrid system ${\mathcal H}$ with $\phi(0,0) =\xi$. $\text{dom}$ stands for domain, and $\setminus $ denotes sets difference.  ${\mathcal O}^{n} = \mathds{R}^{n} \times \{1,-1\}$.
\section{Introduction}

%\subsection{Preliminaries and background}
Informally speaking, a reset controller is any controller, usually referred to as the {\em base controller}, that is equipped with a mechanism for zeroing some of its states, when some event occurs in the control system.  Although the term was coined in the late 90s by Hollot, Chait and coworkers (\cite{Beker04}), specifically to describe "a linear and time invariant system with mechanisms and laws to reset their states to zero", the concept was devised much earlier, in the seminal works of Clegg (\cite{Clegg}) and Horowitz and coworkers (\cite{Krishman74,Horowitz75}). Since then, reset control has considerably evolved by using different resetting laws: the original zero crossing of the error (\cite{Beker04,Banos11,Banos12,Hosse13,Banos16,Hosse21a}), sector-based resetting (\cite{Nesic08,Aangenent09,Nesic11,Tarbouriech11,Loon17}), error band (\cite{Barreiro14,Banos14}), reset at fixed instants (\cite{Guo08,Guo11}), Lyapunov function-based resetting (\cite{Prieur13}), and somehow relaxing the original concept, both including nonlinear base systems and reset to non-zero values in some cases. This has lead to a fecund research area that has been successfully applied in many practical applications, and that has opened many relevant topics in control theory and practice.   

In this work, the focus is on reset control with emphasis on the original concept, using a linear and time-invariant controller that zeroes its   state (fully or partially) when the closed-loop error signal is zero.  The main motivation has been to update and formalize previous work by the authors, that was developed in the the framework of impulsive dynamical systems (IDS), by using the hybrid dynamical systems framework (HI) of \cite{Goebel09,GSTbook}. 
%. The main reason is related with the robustness properties of hybrid systems in this framework, that are specially useful for analyzing reset control systems with a zero-crossing resetting law;
Note that in the IDS framework, resetting laws are based on the exact crossing of the zero error hypersurface, and there is some fragility in detecting a zero crossing specially if measurement noise is present. Although this robustness problem has been alleviated for specific class of exogenous signals (\cite{Banos16}), it is acknowledged the HI framework is more conclusive for equipping reset control systems with good structural properties (specially when considering exogenous signals with jump discontinuities) such as continuous-dependence on initial conditions, closeness of perturbed (due for example to measurement noise) and unperturbed solutions, asymptotic stability is preserved under small perturbations, etc.     

There exist already several relevant works about reset control in the HI framework, most of them based on a sector-based resetting law (see for example \cite{Nesic08,Aangenent09,Nesic11,Tarbouriech11,Loon17}). It is important to emphasize that, in general, this resetting law produces different control solutions in comparison with the zero-crossing resetting law (see  Example 3.4 of this manuscript). Here, it is not argued that one resetting law is superior to the other, as well as in comparison with other resetting laws, they simply are different solutions that may properly work in control practice.

\subsection{Background: Hybrid dynamical systems} %($\cdots$ {\bf cambia a sistemas hibridos con entradas si no utilizas exosistemas} $\cdots$}
This work follows the hybrid system framework developed in \cite{GSTbook} (and references therein), that following \cite{Schutter09}, has been referred to as Hybrid Inclusions (HI) framework, and the reader is referred to \cite{GSTbook,Goebel09} for a detailed exposition of it (see also \cite{CaiTeel09} where hybrid systems with inputs are explicitly analyzed). A hybrid system $\mathcal{H}_{\bf w}$, with state ${\bf x} \in \mathds{R}^n$ and input ${\bf w} \in \mathds{R}^m$, is given by 
\begin{equation}
\mathcal{H}_{\bf w}: \left\{  \begin{array}{llll}
\mathbf{\dot{x}}=f(\mathbf{x},{\bf w}), \quad & (\mathbf{x},\mathbf{w})  \in \mathcal{C}, \\
\mathbf{x}^+=g(\mathbf{x},\mathbf{w}), \quad & (\mathbf{x},\mathbf{w}) \in \mathcal{D}.  
\end{array}
\right.
\label{eq-H}
\end{equation}
and is defined by the following data: i) the {\em flow set} $\mathcal{C}\subset \mathds{R}^n \times \mathds{R}^m$, ii) the {\em flow mapping} $f: \mathds{R}^n\times \mathds{R}^m \rightarrow \mathds{R}^{n}$, 
iii) the {\em jump set} ${\mathcal D} \subset \mathds{R}^n\times \mathds{R}^m$, and iv) the {\em jump mapping} $g:\mathds{R}^n\times \mathds{R}^m \rightarrow \mathds{R}^{n}$. 

Hybrid signals are defined as functions on hybrid time domains (\cite{Goebel06},\cite{GSTbook}). A hybrid arc ${\bf x}:\text{dom } {\bf x} \mapsto \mathds{R}^n$ is a hybrid signal in which ${\bf x}(\cdot,j)$ is locally absolutely continuous for each $j$.  A hybrid input ${\bf w}:\text{dom } {\bf w} \mapsto \mathds{R}^m$ s a hybrid signal in which ${\bf w}(\cdot,j)$ is Lebesgue measurable and locally essentially bounded for each $j$. A solution to \eqref{eq-H} is defined as a pair $({\bf x},{\bf w})$, consisting of a hybrid arc and a hybrid input with $\text{dom } {\bf x}=\text{dom } {\bf w}$ , that satisfy the dynamics of the hybrid system $\mathcal{H}_{\bf w}$ (see \cite{CaiTeel09}-\cite{Sanfelice13} for details about solution pairs to hybrid systems with inputs). % It is clear that, in the definition of solutions pairs to $\mathcal{H}_{\bf w}$, 
 Note that the jump set ${\mathcal D}$ enables jumps but does not force them if there are points in which it is also possible to flow (a similar argument applies to the flow set ${\mathcal C}$); and thus if ${\mathcal C}$ are ${\mathcal D}$ are not disjoint then for a point $(\xi,\psi) \in {\mathcal C}\cap{\mathcal D}$ there may exist several solution pairs $({\bf x},{\bf w})$ to $\mathcal{H}_{\bf w}$ with ${\bf x}(0,0) = \xi$, for any hybrid input ${\bf w}$ with ${\bf w}(0,0) = \psi$.   

For ${\mathcal H}_{\bf w}$, the so-called {\em hybrid basic conditions} defined in \cite{GSTbook}  (see also regularity conditions in \cite{CaiTeel09}) are satisfied if $\mathcal{C}$ and $\mathcal{D}$ are both closed subsets of $\mathds{R}^{n}\times \mathds{R}^{m}$, and if $f$ and $g$ are continuous functions. These hybrid basic conditions guarantee that ${\mathcal H}_{\bf w}$ (without inputs, that is with ${\bf w} =0$) is well posed in the sense that their solution sets inherit several good structural properties: upper-semicontinuous dependence with respect to initial conditions, closeness of perturbed (due for example to measurement noise) and unperturbed solutions, asymptotic stability is preserved under small enough perturbations (\cite{GSTbook}).

\subsection{Organization of the manuscript}
This work is mainly devoted to develop a representation of reset control systems, with a zero-crossing resetting law and a linear and time invariant base system, in the HI framework. The focus is about well-posedness and stability, with a strong motivation in obtaining HI models that capture key properties in control practice. In Section 2, starting with a new Clegg integrator model equipped with an input zero-crossing detection (ZCD) mechanism, it is postulated a reset controller model based on it. Section 3 analyzes closed-loop reset control systems, resulting of the feedback connection of a reset controller (with the ZCD mechanism) and a linear and time invariant plant (using plant output measurement). Some basic properties of closed-loop hybrid system like well-posedness, existence of solutions and flow persistence (a concept introduced to guaranty the existence of solutions that are unbounded in the $t$-direction) are investigated. Also, a deep analysis of how solutions to the closed-loop hybrid system are related to operation in control practice. To avoid existence of defective solutions, a standard approach based on time regularization is used. Finally, a reset control system in the HI framework, with a zero-crossing resetting law and time-regularized, is postulated. An example, based on a classical case analyzed by Horowitz, is investigated with the proposed model; also a comparison with a time-regularized reset control system with a sector-based resetting law is performed. In Section 4, stability of the proposed reset control system is investigated. A basic result will be a reformulation of previous stability result in the new HI framework, relating stability of the closed-loop hybrid system with the stability of a discrete-time system obtained as a Poincar\'e-like map. Two different stability approaches are then investigated: one based on the analysis of the reset interval sequences periods, that result in analyzing eigenvalues of different matrices associated with those periods; and another based on the use of Lyapunov functions that finally results in LMIs conditions whose feasibility determine the stability of the reset control system. Moreover, several examples, that illustrates the applicability of the proposed results, are developed.

\section{From the Clegg integrator to reset controllers}

%In the HI framework, robustness to measurement noise is developed as a particular case of robustness to  perturbations. Although this property has been more or less explicitly analyzed in several works (for example in \cite{Goebel09} noise is explicitly embedded in a more general outer perturbation), to the authors knowledge a formal definition of the property is missing. In the following, we define the robustness to measurement noise in the spirit of the HI framework with the goal of developing a precise analysis in the next section. 

In this work, the main focus is on reset control systems in which a continuous-time plant is controlled by a reset controller with plant output feedback (see Fig. \ref{fig:controlsetup}). 
This feedback control system, that uses plant output measurement, will be modeled in the HI framework by using \eqref{eq-H}. More specifically, the plant is linear and time invariant (LTI) and single-input single-output, and described by the differential equation: 
\begin{equation}
P:\left \{
\begin{array}{l}
\dot{\mathbf{x}}_p = A_p\mathbf{x}_p + B_p{u} \\
y_p = C_p\mathbf{x}_p
\end{array}
\right.
\label{eq:P}
\end{equation}
where $\mathbf{x}_p\in \mathds{R}^{n_p}$ is the plant state, ${u} \in \mathds{R}$ the control input, $y_p \in \mathds{R}$ is the plant output, and  $A_p$, $B_p$ and $C_p$ are matrices of appropriate dimensions. The reset controller, with continuous state $\mathbf{x}_r\in \mathds{R}^{n_r}$, will be endowed with a zero-crossing detection mechanism based on a discrete state $q \in \{1,-1\}$, being finally $(\mathbf{x}_r,q) \in {\mathcal O}^{n_r} := \mathds{R}^{n_r} \times \{1,-1\}$ the controller state. In the following, the proposed reset controller will be analyzed in detail; since the controller setup will be based on a modification of the Clegg integrator (\cite{Clegg}), this is first described.  
\begin{figure}[t] %  figure placement: here, top, bottom, or page
  % \small
   \centering
    \includegraphics[width=5cm]{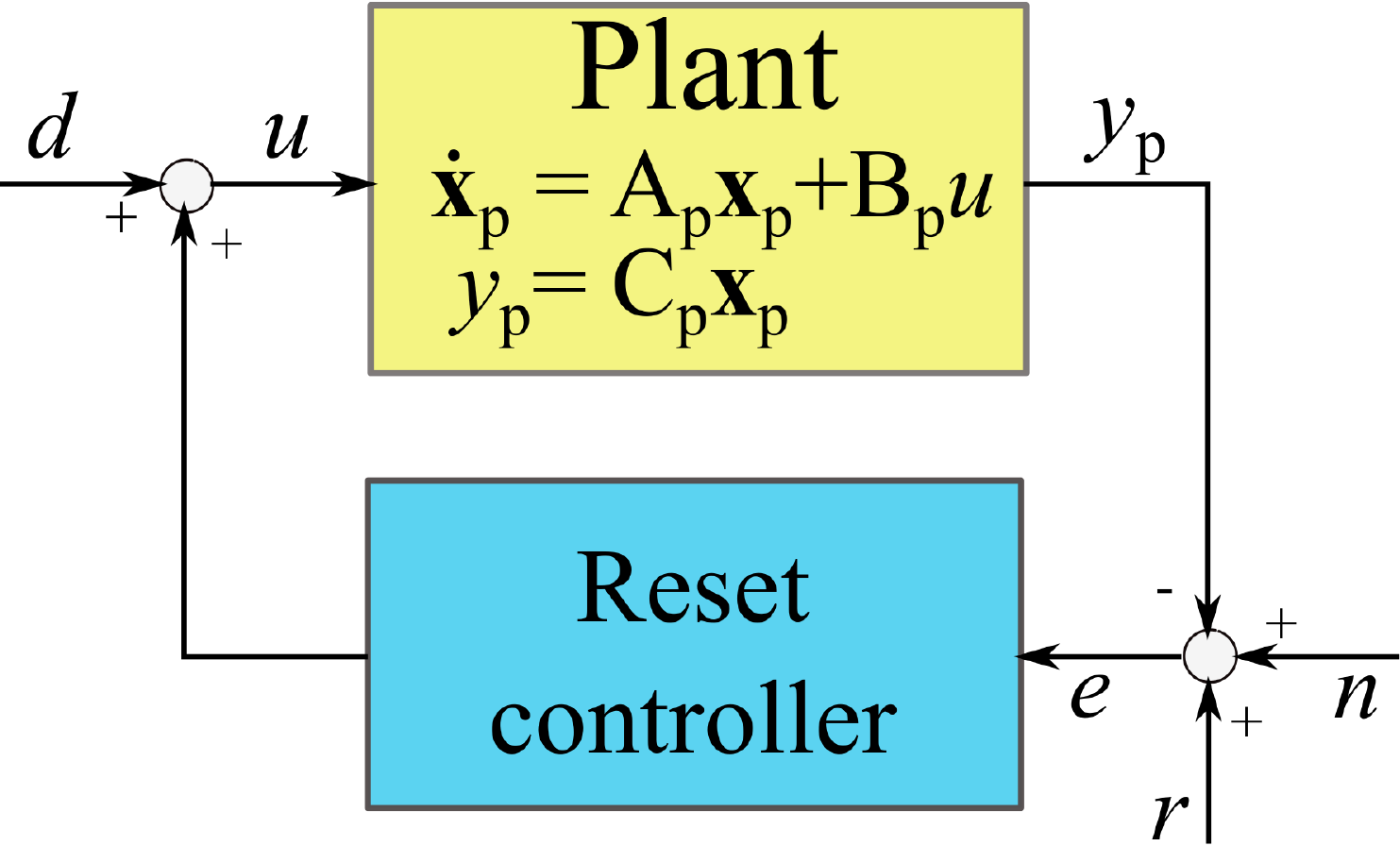}
    \caption{A reset control system, with a LTI continuous-time plant and a feedback reset controller. The feedback loop is perturbed by errors both in the measurement (noise $n$) and the actuator (disturbance $d$); $r$ is a reference signal. %The perturbed closed-loop hybrid system $\mathcal{H}_{(\mathbf{d}_1,\mathbf{d}_2)}^\text{cl}$ is given by (5), while the (unperturbed) hybrid control system $\mathcal{H}^\text{cl}$ is obtained by making $\mathbf{d}_1 = \mathbf{d_2} = \mathbf{0}$ in (5). Note these hybrid systems models are particular cases of $\mathcal{H}_{(\mathbf{e}_1,\mathbf{e}_2,\mathbf{e}_3)}$ and $\mathcal{H}$, as given by \eqref{eq-He} and \eqref{eq-H}, respectively.
    }
    \label{fig:controlsetup}
\end{figure}

\subsection{A Clegg integrator wih a zero-crossing detection mechanism} %and spatial regularization}
A basic and well-known reset controller is the Clegg integrator (\cite{Clegg},\cite{Krishman74}), that will be adapted in this work attaching a zero-crossing detection procedure based on a the discrete state $q \in \{1,-1\}$, and also adding an extra input.  Besides the trigger input $e \in \mathbb{R}$ (usually the error signal in the case of an output feedback control system) the input signal $e_{CI} \in \mathbb{R}$ is proposed\footnote{Note that the original Clegg integrator is recovered from $\text{CI}$ by removing the discrete state $q$ and doing  $e_{CI} = e$.}. 
Using \eqref{eq-H}, the result is a new model of the Clegg integrator in the HI framework. %, that will be referred to as $\text{CI}$ (with some abuse of notation)
 It is given by: 
\begin{equation}
\text{CI} : 
\left \{
\begin{array}{ll}
\begin{array}{l}
  \dot{x}_r    
\end{array} 
= 
\begin{array}{l}
  e_{CI}  
\end{array} 
 &\text{, } (x_r,q,e_{CI},e)  \in \mathcal{C} \\
\\
\left(
\begin{array}{l}
     {x}_r^+ \\
     {q}^+
\end{array}  \right) =
\left(
\begin{array}{cc}
  0& 0\\
  0 &-1
\end{array} \right)
\left(
\begin{array}{l}
   {x}_r\\
   q  
\end{array} \right)
&\text{, } (x_r,q,e_{CI},e)  \in \mathcal{D} 
\end{array} 
\right.
\label{eq:CI}
\end{equation}
where $(x_r,q) \in \mathcal{O}$ is the CI state, $(e_{CI},e) \in \mathds{R}^2$ is its input, and $v = x_r$ is its output, and the flow set $\mathcal{C}$ and the jump set $\mathcal{D}$ are given by 
\begin{equation}
\mathcal{C}  = \{(x_r,q,e_{CI},e) \in \mathcal{O}\times \mathds{R}^2: qe \leq0  \}
\label{eq:C-CI}
\end{equation}
and 
\begin{equation}
\mathcal{D} = \{(x_r,q,e_{CI},e) \in \mathcal{O}\times \mathds{R}^2: qe \geq0  \},
\label{eq:D-CI}
\end{equation}
respectively.  %From their definitions, it directy follows that both $\mathcal{C}$ and $\mathcal{D}$ are closed sets in $ \mathcal{O}$  and that $\mathcal{C} \cup \mathcal{D} =  \mathcal{O}$. 
Note that since the CI discrete state $q$ is constant when flowing, its flow equation is not explicitly shown by simplicity. The two input signals $e_{CI}$ and $e$ are useful for modeling more complex reset controllers by using $\text{CI}$ as a building block (see Fig. \ref{fig:CIyFORE}). This capability will be fully exploited by higher order reset controller in the next Section.

\begin{figure}[h]
\centering
{\includegraphics[width=0.55\textwidth]{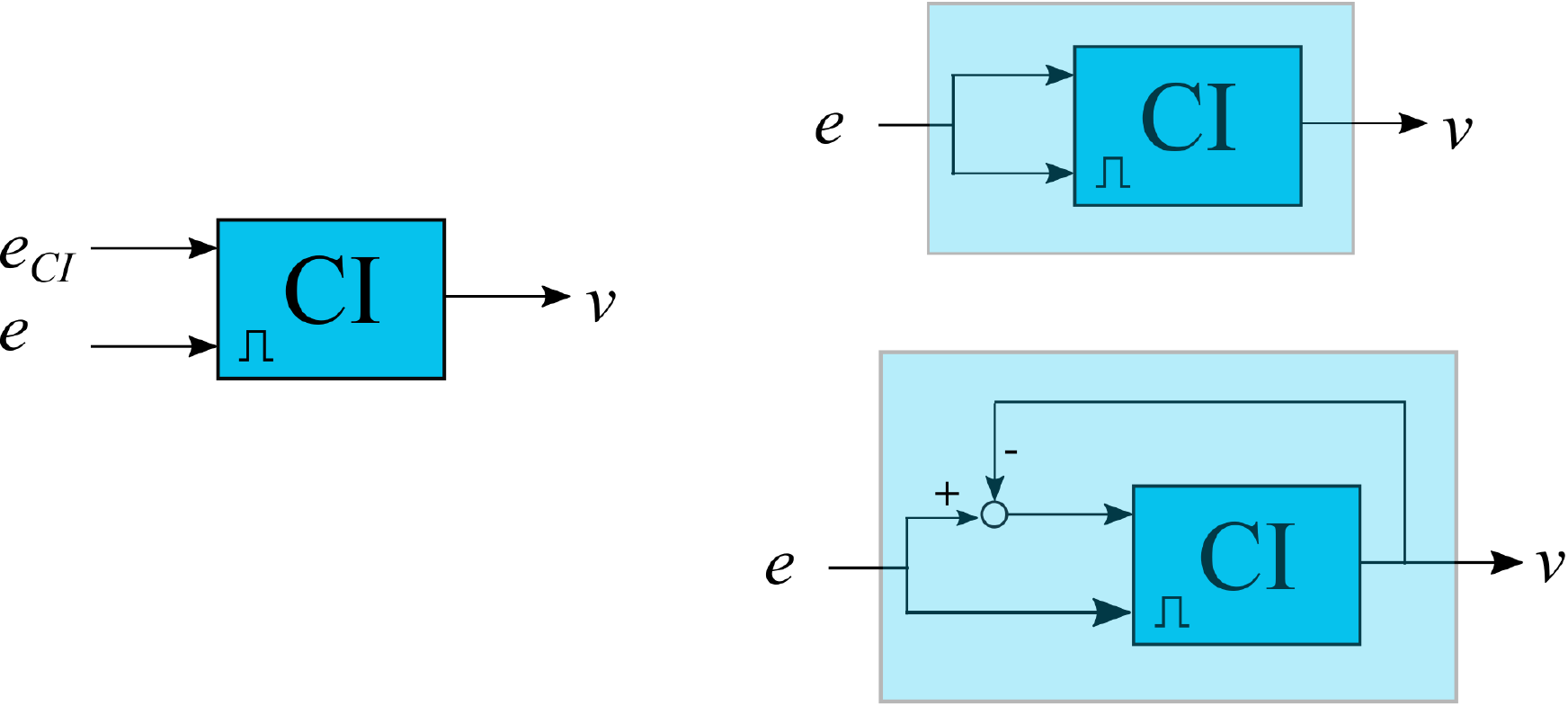}}
\caption{ ({\em left}) The block $CI$: representation of the Clegg integrator with two inputs and a zero-crossing detection; ({\em right}) The Clegg integrator with one input ({\em top}), and a first order reset element (FORE) ({\em bottom}), built upon the block CI.}.
\label{fig:CIyFORE}
\end{figure}

 %when there are two or more $\text{CI}_\text{ZCD}$ blocks and it is wanted that all of them perform jumps at the same reset instants (this case will be treated in the next section). If the reset controller only has a $\text{CI}_\text{ZCD}$ block, the trigger signal would be the continuous $\text{CI}_\text{ZCD}$ state, that is $v_r = x_r$. This case is analyzed in the following. 
%The flow and jump regions in the $(e,v)$-plane are defined by $F_q = \{(e,x_r)\in \mathbb{R}^2: (e,x_r,q)\in \mathcal{C}\}$ and $J_q = \{(e,x_r)\in \mathbb{R}^2: (e,x_r,q)\in \mathcal{D}\}$, respectively. 

%For example, If the reset controller is a Clegg integrator with one input (Fig. \ref{fig:CIyFORE}-right-top), then $e_{CI} = e$ in 

Moreover, subsets $\mathcal{C}_1 \subset \mathcal{C}$ and $\mathcal{D}_1 \subset \mathcal{D}$ are defined as $\mathcal{C}_1  = \{(x_r,q,e_{CI},e) \in \mathcal{C}: q = 1\}$ and $\mathcal{D}_1  = \{(x_r,q,e_{CI},e)  \in \mathcal{D}: q=1  \}$; the subsets $\mathcal{C}_{-1} \subset \mathcal{C}$ and $\mathcal{D}_{-1} \subset \mathcal{D}$ are defined accordingly. When $(x_r,q,e_{CI},e)$ goes from ${\mathcal {C}_{\pm 1}}$ to $\mathcal{D}_{\pm 1}$  either crossing or jumping through their boundary, a jump of the $\text{CI}$ state may be performed. This guarantees the detection of a zero-crossing even if the signal $e$ has jump discontinuities, for example due to some noise measurement $n$ (see Fig. \ref{fig:CI-ZCD}).  
%, a jump may be performed either crossing or jumping through their boundary ${F_{\pm 1}} \cap {J_{\pm 1}}$) and a jump is performed, then after the jump $q^+=-q$ and  
%$(e,v_r^+) \in {\mathcal D}_{\pm 1}$
%the $(e,v_r)$-plane is partitioned into two closed sets (see Fig. \ref{fig:CI-ZCD}), and the discrete state $q$ changes when the continuous state $x_r$ goes from the region ${\mathcal C}_{\pm 1}$ to the region ${\mathcal D}_{\pm 1}$  (either crossing or jumping through their boundary). %\footnote{Here, obviousy zero-crossing means that the signal $e$ has a small value since $\varepsilon$ is assumed to be small.}. 
%The flow set is $\mathcal{C}_e = \{(x_c,q) \in \mathds{R}^2 : q(e- \varepsilon x_c)\leq0   \} $, and the jump set is $\mathcal{D}_e = \{(x_c,u) \in \mathds{R}^2 : q(e- \varepsilon x_c)\ \geq 0   \} $. 

\begin{figure}[h]
\centering
{\includegraphics[width=0.4\textwidth]{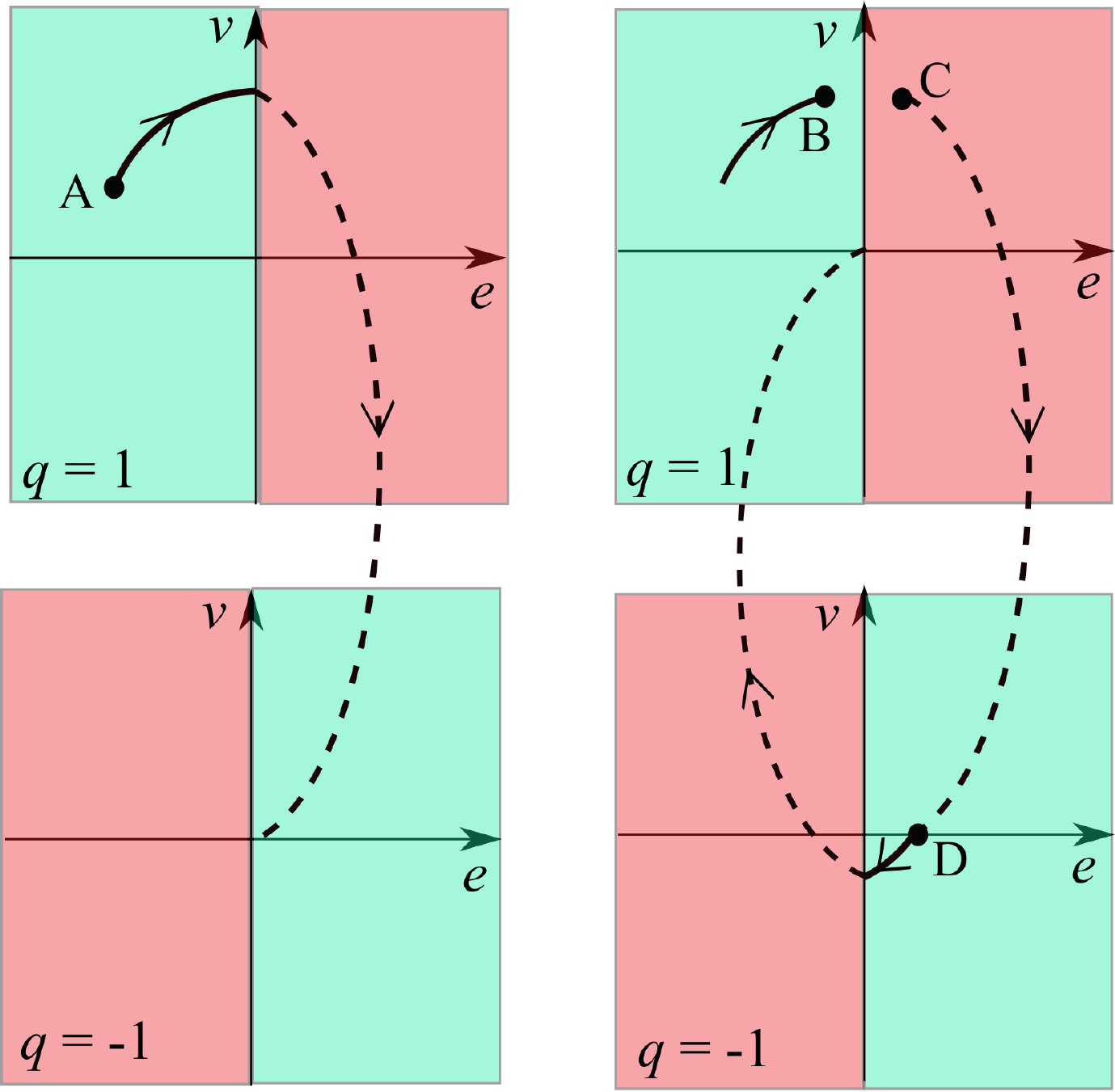}}
		\caption{ Zero-crossing detection mechanism: ({\em left}) from the initial point A ({\em up}), the system flows until a zero-crossing is detected ($qe = 0$), jumping from $\mathcal{C}_1 \cap \mathcal{D}_1$ to $(e,v,q) = (0,0,-1)$ ({\em bottom}); ({\em right}) a perturbation of the error signal at point B makes the system jump from B$\text{ } \in\mathcal{C}_1$ to C$\text{ } \in \mathcal{D}_1$  ({\em up}), then a zero-crossing is also detected ($qe>0$) and the system jumps to D$\text{ }\in \mathcal{C}_{-1} \setminus \mathcal{D}_{-1}$  ({\em bottom}); from D the system flows again and finally jumps to $(0,0,1)$ ({\em up}).}
\label{fig:CI-ZCD}
\end{figure}

%$\text{ } \in\mathcal{C}_1$ to C $\text{ } \in \mathcal{D}_1$ 

%One advantage of the proposed $\text{CI}_\text{ZCD}$ model, that basically relies on the spatial regularization, is that, with the exception of a zero input and a zero state (that is $(e,x_r) = (0,0))$, if a jump is performed then the state after the jump only can flow, and it will flow for some finite (continuous) time interval. This property, which is key in practice regarding controller implementation, will be referred to as {\em flow-persistence}. More precisely, flow persistence means that if $(e,x_r,q)  \in \mathcal{D}$ and $(e,x_r) \neq (0,0)$, then $(e,x^+_r,q^+) \in \mathcal{C}\setminus \mathcal{D}$. 
%Note that for  $\text{CI}_\text{ZCD}$, flow-persistence directly follows: from \eqref{eq:D-CI} it results that $qe >0$ since $e\neq0$, and thus $q^+e  = -qe < 0$, that is $(e,x^+_r,q^+) \notin \mathcal{D}$; and in addition, since according to \eqref{eq:C-CI}, $\mathcal{C}$ contains the relative complement of $\mathcal{D}$ in  $\mathcal{O}^2$, then $(e,x^+_r,q^+) \in \mathcal{C}\setminus \mathcal{D}$.

\vspace{0.5cm}

\subsection{A Reset controller with zero-crossing resetting law}
In this work, it is proposed a reset controller inspired in \cite{Beker04,BBbook12,Banos16}. Here, the new $\text{CI}$ is used as a building block, and thus the reset controller has also attached a zero-crossing detection mechanism. It has a state $(\mathbf{x}_r,q) \in \mathcal{O}^{n_r}$ and a scalar input $e$. Using again \eqref{eq-H}, it is given by
\begin{equation}R:
\left \{
\begin{array}{ll}
\begin{array}{l}
     \dot{\mathbf{x}}_r  
\end{array} 
=
\begin{array}{l}
   A_r {\mathbf x}_r  
\end{array}  
+
\begin{array}{l}
  B_r e
 \end{array} 
 &\text{, } ({\mathbf x}_r,q,e) \in \mathcal{C} \\
\\
\left(
\begin{array}{l}
     {\mathbf x}_r^+ \\
     {q}^+
\end{array}  \right) =
\left(
\begin{array}{cc}
  A_\rho& 0\\
  0 & -1
\end{array} \right)
\left(
\begin{array}{l}
    {\mathbf x}_r \\
    q 
\end{array} \right)
&\text{, } ({\mathbf x}_r,q,e) \in \mathcal{D}
\end{array} 
\right.
\label{eq:ResetController}
\end{equation}
where the output and control signal is the scalar $v = C_r \mathbf{x_r} + D_r e$, and now the flow and jump sets, $\mathcal{C}$ and $\mathcal{D}$, are  
\begin{equation}
\mathcal{C}  = \{({\mathbf x}_r,q,e) \in \mathcal{O}^{n_r}\times \mathds{R}: qe \leq0  \}\end{equation}
and 
\begin{equation}
\mathcal{D} = \{({\mathbf x}_r,q,e) \in \mathcal{O}^{n_r}\times \mathds{R}: qe \geq0  \},
\label{eq:D-CI}
\end{equation}
%by \eqref{eq:C-CI} and \eqref{eq:D-CI},
respectively. Here $A_r$, $B_r$, $C_r$, and $D_r$ are real matrices with appropriate dimensions. And $A_\rho$ is a matrix that set to zero some of the controller states when a zero-crossing has been detected (by convention, the last $n_\rho$ states of $\mathbf{x}_r$ are set to zero, while the first $n_{\bar{\rho}} = n_r - n_\rho$ states are kept without change). It is given by 
\begin{equation}%\small
A_{\rho}= 
\left (
\begin{array}{cc}
I_{n_{\bar{\rho}} \times n_{\bar{\rho}} } & 0_{n_{\bar{\rho}} \times {n_{\rho}}%
}\\
0_{n_{\rho} \times n_{\bar{\rho}} } & 0_{n_{\rho} \times n_{\rho} }
\end{array}
\right ),
\label{eq:AR}
\end{equation}
%where $n_{\bar{\rho}} = n - n_{\rho}$.
In the case of a \emph{full reset} controller $n_{\rho} = n_r$, while if $n_{\rho}<n_r$ then $R$ is a \emph{partial reset} controller.  In addition, $A_r$, $B_r$, and $C_r$ are partitioned into blocks with appropriate block dimensions:
\begin{equation}%\small
   A_r =  \left (
   \begin{array}{cc}
      A_{r_{11}} &\hspace{-0.12cm} A_{r_{12}}\\
      A_{{r_{21}}} & \hspace{-0.12cm}A_{r_{22}}
   \end{array}
   \right), 
   B_r =  \left (
   \begin{array}{c}
      B_{r_{1}}\\
      B_{r_{2}}
   \end{array}
   \right), 
   C_r =  \left (
   \begin{array}{cc}
      C_{r_{1}} & \hspace{-0.12cm}  C_{r_{2}}
   \end{array}
   \right) 
 % \label{(5.10)}
\end{equation}

%\begin{figure}[t]
%\centerline{\includegraphics[width=9cm,keepaspectratio]{graf/CI2.pdf}}
%\caption{({\em left}) Two-inputs Clegg integrator, ({\em right}) Clegg integrator (by simplicity, the connection of the signal $e$ to the trigger input means that $s$ is the bolean expression $e = 0 \land v \neq 0$).} 
%\label{CI}
%\end{figure}

\begin{figure}[ht]
\centering
{\includegraphics[width=0.55\textwidth]{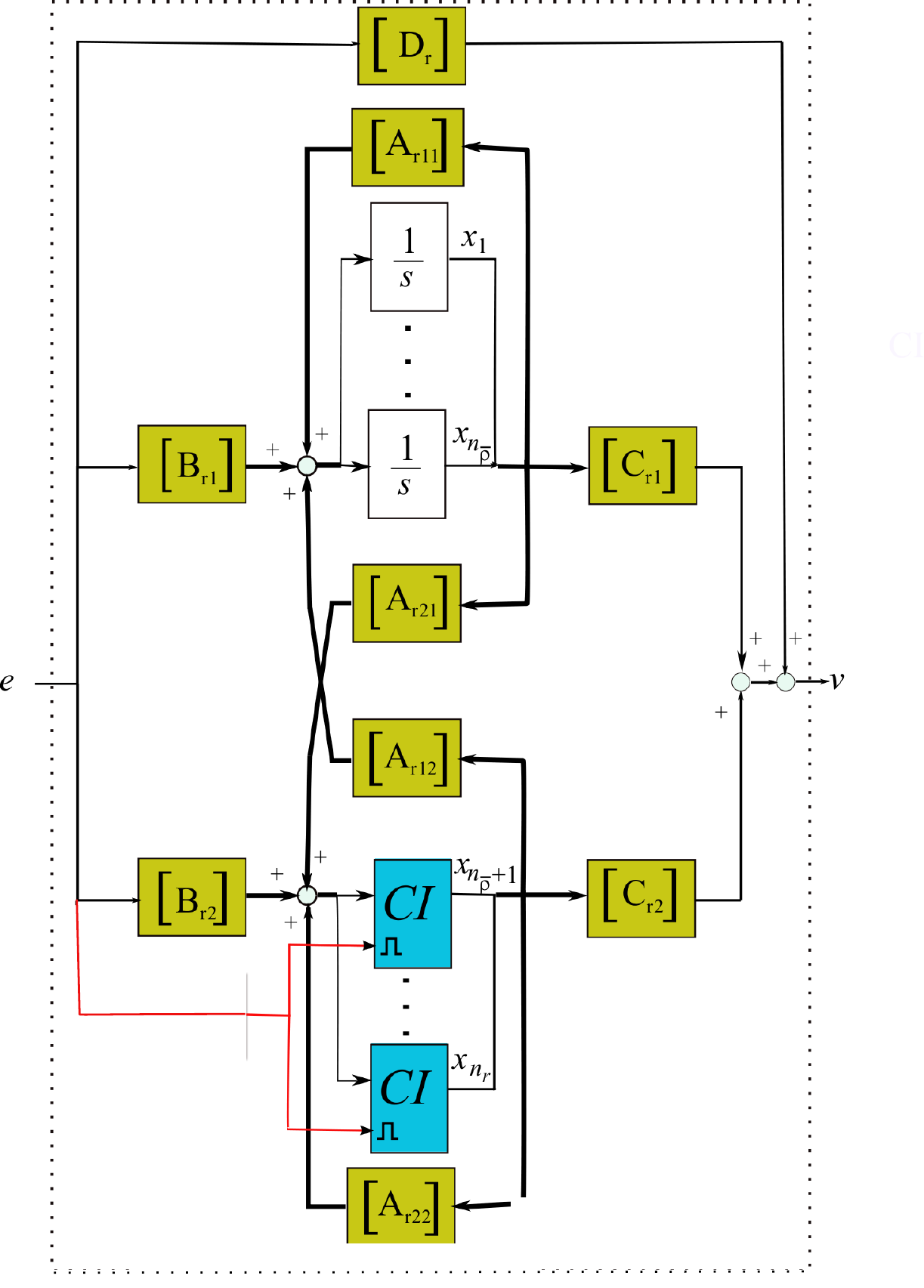}}
\caption{A block diagram of the reset controller with a zero-crossing resetting law, given by \eqref{eq:ResetController},  with state $({\bf x}_r,q)$. Here ${\bf x}_r = (x_1, \cdots,x_{n_{\bar{\rho}}}, x_{n_{\bar{\rho}}+1}, \cdots,x_{n_r})$ and $q$ is the discrete state that is common for all the CI blocks. For a block [M] the output vector is the matrix multiplication of M with the input vector (thin lines correspond to scalar signals, while thick lines correspond to vector signals).
%The CI blocks have the same discrete state $q$
}
\label{fig:ResetController}
\end{figure}

For a block diagram representation of the reset controller $R$, besides integrator blocks it is sufficient to use the modified Clegg integrator $\text{CI}$ given by \eqref{eq:CI} as a basic block.  %as shown in Fig. \ref{CI}
A block diagram of $R$ that allows a direct implementation is given in Fig. \ref{fig:ResetController}. 
On the other hand, if $A_{r_{21}} = O$ ($A_{r_{12}} = O$) then $R$ will be referred to as a {\em right reset controller} ({\em left reset controller}); 
the name is related with the right (left) triangular block structure of the matrix $A_r$. Informally speaking, for a right reset controller the inputs of the CI blocks are not influenced by the outputs of the integrator blocks. It is worthwhile to mention that some of  the reset controller with partial reset (see Fig. \ref{fig:RRC}) that has been found useful in practice are right reset controllers (\cite{BBbook12}).

\begin{figure}[t]
\centerline{\includegraphics[width=7cm,keepaspectratio]{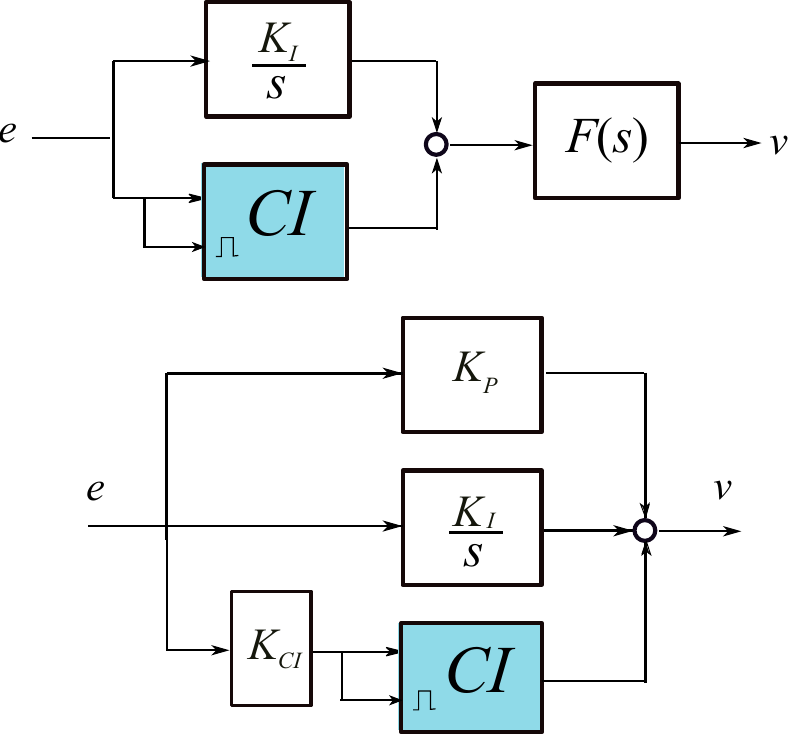}}
\caption{Right-reset controllers: ({\em top}) Horowitz reset controller (\cite{Krishman74}, $K_I$ is a parameter to be tuned and $F(s)$ a transfer function to be designed), ({\em botton}) PI+CI controller (\cite{BBbook12}, $K_P$, $K_I$ and $K_{CI}$ are controller parameters to be tuned).}
\label{fig:RRC}
\end{figure}

\section{The closed-loop reset control system} %${\mathcal H}^{cl}_{\bf w}$}
Once the plant and the reset controller has been defined, now the feedback control system  ${\mathcal H}^{cl}_{\bf w}$ is obtained (see Fig. \ref{fig:controlsetup}) as a hybrid control system, where the plant output $y_p$ is the feedback signal. Its state is  $(\mathbf{x}_p,\mathbf{x}_r,q) \in \mathcal{O}^{n}$, %where $n = n_p+n_r$,
 and is given by %By simplicity, firstly the case of no exogenous signal is analyzed, but exogenous signals like the reference $r$, measurement noise $n$ and actuator disturbances $d$ will be considered later. 
\begin{equation} {\mathcal H}^{cl}_{\bf w}:
\left \{
\begin{array}{ll}
%\begin{array}{l}
%  \dot{\mathbf{x}}  
%\end{array}
\left(
\begin{array}{l}
  \dot{\mathbf{x}}_p \\
  \dot{\mathbf{x}}_r 
\end{array}  \right)
=
\begin{array}{l}
  \left(
\begin{array}[c]{cccc}%
 A_p - B_pD_rC_p & B_pC_r\\
  -B_rC_p & A_r  %
\end{array}
\right) 
%{\mathbf x}
\left(
\begin{array}{l}
  {\mathbf{x}}_p \\
  {\mathbf{x}}_r 
\end{array}  \right)
\end{array}  
+ 
\left(
\begin{array}[c]{cccc}%
 B_pD_r & B_p\\
  B_r & O  %
\end{array}
\right) {\mathbf w}
&\text{, } (\mathbf{x},q,{\bf w}) \in \mathcal{C}_{\mathbf w}^{cl} \\
\\
\left(
\begin{array}{l}
  {\mathbf x}_r^+\\
  {q}^+
\end{array}  \right) =
\left(
\begin{array}{cc}
  A_\rho& 0\\
  0 &-1
\end{array} \right)
\left(
\begin{array}{l}
   {\mathbf x}_r\\
   {q}
\end{array} \right)
&\text{, } (\mathbf{x},q,{\bf w}) \in \mathcal{D}_{\mathbf w}^{cl} 
\end{array} 
\right.
\label{eq:Hclw}
\end{equation}
where the exogenous input is  ${\bf w}= (w_1,w_2) = (r+n,d)$, and the flow and jump sets are given by  
\begin{equation}
\mathcal{C}_{\mathbf w}^{cl}  = \{ (\mathbf{x}_p,\mathbf{x}_r,q,{\bf w}) \in \mathcal{O}^{n_p+n_r} \times \mathds{R}^2: q(w_1 - C_p{\bf x}_p) \leq 0
\} \label{eq:Cclw}
\end{equation}
and
\begin{equation}
\mathcal{D}_{\mathbf w}^{cl}  = \{ (\mathbf{x}_p,\mathbf{x}_r,q,{\bf w}) \in \mathcal{O}^{n_p+n_r}  \times \mathds{R}^2: q(w_1 - C_p{\bf x}_p) \geq 0
\},\label{eq:Dclw}
\end{equation}
respectively.

%\subsection{Exogenous signals}

\subsection{Closed-loop solutions and properties} %for Bohl exogenous signals}
Following Section 1.2, solutions to ${\mathcal H}^{cl}_{\bf w}$ are defined as pairs $({\bf x},{\bf w})$ satisfying \eqref{eq:Hclw}, where 
${\bf x}$ is a hybrid arc and ${\bf w}$ is a hybrid input. Since ${\mathcal H}^{cl}_{\bf w}$ satisfies the hybrid basic conditions (the flow and jump maps are continuous, and the flow and the jump sets are closed), it directly follows that for the case of no exogenous inputs, that is for ${\bf w} = {\bf 0}$, the system ${\mathcal H}^{cl}_{\bf w}$ is well posed and the property of asymptotic stability is robust (see \cite{GSTbook} for precise definitions and results). 

Now, a key question is to analyze if ${\mathcal H}^{cl}_{\bf w}$ has also good properties for relevant sets of exogenous inputs. Firstly, the problem of modeling relevant exogenous inputs arises; since the plant is originally continuous-time, it is natural to consider that exogenous signals only depend on time $t$ and not on $j$, and thus hybrid inputs ${\bf w}$ can be considered as given by setting ${\bf w}(t,j) = {\bf w}'(t)$ for all $(t,j) \in E$, for some continuous time signal ${\bf w}'$ and any arbitrary time domain $E$. Moreover, since exogenous inputs should be relevant in control practice, it will be assumed that ${\bf w}'$ is generated by an exosystem $\Sigma$ with state ${\mathbf x}_{w}$, that is

\begin{equation} {\Sigma}:
\left \{
\begin{array}{l}
\begin{array}{l}
  \dot{\mathbf{x}}_w  
\end{array}
=
\begin{array}{l}
  A_w {\mathbf x}_w \end{array}  
\\
{\mathbf w} =
\left(
\begin{array}{c}
  C_{w_1}\\
  C_{w_2} 
\end{array} \right)
{\mathbf x}_w
\end{array} 
\right.
\label{eq:ExoS}
\end{equation}

These exosystems will allow to generate signals like steps, ramps, sinusoids, and in general Bohl functions. Finally, these exosystems are embedded in the reset control system \eqref{eq:Hclw} resulting in the (autonomous) reset control system ${\mathcal H}^{cl}$ with state $ ({\mathbf x},q)=({\mathbf x}_{w},{\mathbf x}_p,{\mathbf x}_r,q)$:
\begin{equation} {\mathcal H}^{cl}:
\left \{
\begin{array}{ll}
\begin{array}{l}
  \dot{\mathbf{x}}  
\end{array}
=
\begin{array}{l}
  A {\mathbf x} \end{array}  
&\text{, } ({\mathbf x},q) \in \mathcal{C}^{cl} \\
\\
\left(
\begin{array}{l}
  {\mathbf x}^+\\
  {q}^+
\end{array}  \right) =
\left(
\begin{array}{cc}
  A_R& 0\\
  0 &-1
\end{array} \right)
\left(
\begin{array}{l}
   {\mathbf x}\\
   {q}
\end{array} \right)&\text{, } ({\mathbf x},q) \in \mathcal{D}^{cl} 
\end{array} 
\right.
\label{eq:Hcl}
\end{equation}
where the matrices $A$ and $A_R$ are given by 
\begin{equation}
\begin{array}{ccc}
 {A}=\left(
\begin{array}[c]{ccc}%
A_w&O&O\\
 B_p(D_r C_{w_1}+C_{w_2})&A_p - B_pD_rC_p & B_pC_r\\
  B_rC_{w_1}&-B_rC_p & A_r  %
\end{array}
\right),  &
{A_R}=\left(
\begin{array}[c]{cccc}%
I&O&O\\
O& I & O\\
 O& O& A_\rho  %
\end{array}
\right), 
 \end{array} \label{eq:AAR}
 \end{equation}
and the flow and jump sets are given by  
\begin{equation}
\mathcal{C}^{cl}  = \{ ({\mathbf x},q) \in \mathcal{O}^{n}: qC{\bf x} \leq 0
\} \label{eq:Ccl}
\end{equation}
and
\begin{equation}
\mathcal{D}^{cl}  = \{ ({\mathbf x},q) \in \mathcal{O}^{n}: qC{\bf x} \geq 0
\},\label{eq:Dcl}
\end{equation}
respectively, where $n = n_w+n_p+n_r$, and $C$ is given by 
\begin{equation}
C = 
\left(
\begin{array}{ccc}
 C_{w_1} & -C_p   & O 
\end{array}
\right)
\label{eq:C}
\end{equation}

The next proposition analyzes the existence of solutions to ${\mathcal H}^{cl}$ and some properties that will be useful in control practice. For definitions of well-posedness and Zeno solutions see \cite{GSTbook}.  ${\mathcal H}^{cl}$ is {\em flow persistent} if for any $\xi \in \mathcal{O}^n$ there exist a solution $\phi$ to ${\mathcal H}^{cl}$ with $\phi(0,0) = \xi$, such as $\text{dom }\phi$ is unbounded in the $t$-direction, that is the set $\{ t \in \mathbb{R}_{\ge 0}: \exists j \in \mathbb{N} \text{ such that } (t,j) \in \text{dom }\phi \}$ is not upper-bounded. 

\begin{proposition}
\label{th:existSolution}
Consider the reset control system ${\mathcal H}^{cl}$, and a point $\xi = (\mathbf{x}_0,q_0) \in \mathcal{O}^n$, then:
\begin{enumerate}

%\item It is a homogeneous hybrid dynamical system.
\item (Well-posedness)  ${\mathcal H}^{cl}$ is well-posed. %satisfies the basic hybrid conditions.
\item (Existence of solutions) %For any $\xi \in \mathcal{O}^n$
There exist nontrivial solutions $\phi$ to ${\mathcal H}^{cl}$ with $\phi(0,0) = \xi$, and if $\phi \in {\mathcal S}_{{\mathcal H}^{cl}}(\xi)$ then it is complete, that is ${\mathcal H}^{cl}$ is forward complete from $\mathcal{O}^n$. 
\item (Flow persistence) ${\mathcal H}^{cl}$ is flow persistent.

\end{enumerate}
\end{proposition}
\begin{proof}
\begin{enumerate}
%\item The ZCD-reset system is homogeneous with degree 0 with respect to the standard dilation $M_\lambda := \lambda I_{(n+1)\times (n+1)}$. It directly follows from the fact that $F_{ZCD}(M_\lambda (\mathbf{x},q)) = (\lambda A \mathbf{x},0)=M_\lambda F_{ZCD}((\mathbf{x},q))$, and since $\lambda > 0$ then $G_{ZCD}(M_\lambda (\mathbf{x},z)) = (\lambda A_R\mathbf{x},-\lambda q)= M_\lambda G_{ZCD}((\mathbf{x},q))$. And, on the other hand, for the flow set $M_\lambda \mathcal{C}_{ZCD} = \{ (\mathbf{y},p) \in \mathds{R}^{n+1}: \mathbf{y}= \lambda \mathbf{x}, p = \lambda q, \lambda \in \mathds{R}, (\mathbf{x},q) \in \mathcal{C}_{ZCD} \}$, since $\lambda >0$, $qC\mathbf{x} \geq 0$ and $pC\mathbf{y}  = \lambda^2 qC\mathbf{x} \geq 0$  it directly follows that 
%$M_\lambda \mathcal{C}_{ZCD} =  \mathcal{C}_{ZCD}$; for the jump set a similar argument shows that  $M_\lambda \mathcal{D}_{ZCD} =  \mathcal{D}_{ZCD}$. 
%%
\item It directly follows since  ${\mathcal H}^{cl}$ satisfies the basic hybrid conditions. Note that $\mathcal{C}$ and $\mathcal{D}$ are closed, and the functions $f,g:\mathcal{O}^n \rightarrow \mathcal{O}^n$, defining the flowing and jumping dynamics, respectively, and given by  
$f(({\mathbf x},q))= (A\mathbf{x},0)$ and $g(({\mathbf x},q))= (A_R{\mathbf x},-q)$, are continuous.
\item If $\xi \in \mathcal{C} \setminus \mathcal{D}$ then there exists a solution $\phi$ with $\phi(t,0) = (e^{At}\mathbf{x}_0,q_0) \in \mathcal{C}$ for $t \in [0,\epsilon]$ and $\epsilon = \min\{t \in \mathds{R}^+ : Ce^{At}\mathbf{x}_0 = 0\}$; in the case that $Ce^{At}\mathbf{x}_0 \neq 0$ for any $t > 0$ then $\phi(t,0) = (e^{At}\mathbf{x}_0,q_0) \in \mathcal{C}$ for $t \in [0,\infty)$. Also, if $\xi \in  \mathcal{D}$ then there exists solutions $\phi$ with $\phi(0,1) = (A_R {\mathbf x}_0,-q_0)$ and $\phi(0,0) = \xi$. Thus there exists a nontrivial solution to ${\mathcal H}^{cl}$ starting from $\xi$. Moreover, for $({\mathbf x},q) \in {\mathcal D}^{cl}$ it is true that $qC{\mathbf x}\geq 0$, and thus after a jump $q^+C{\mathbf x}^+ = -qCA_R{\mathbf x} = -qC{\mathbf x} \leq 0$, that is $({\mathbf x}^+,q^+) \in {\mathcal C}^{cl}$ and thus $g({\mathcal D}^{cl}) \subset {\mathcal C}^{cl}$; since, in addition, any solution to the flow equation $\dot{{\mathbf x}} = A {\mathbf x}$  defined on an interval, open to the right, can be trivially extended to an interval including the right endpoint, it is concluded that any maximal solution is complete (\cite{GSTbook}, Prop. 2.10). 
\item %Since $g({\mathcal D}^{cl}) \subset {\mathcal C}^{cl}$ then is may be assumed that system solutions starts at ${\mathcal C}^{cl}$. 
The only obstacle for the existence of solutions that are unbounded in the $t$-direction is that when the system jumps from ${\mathcal C}^{cl}\cap{\mathcal D}^{cl}$ to ${\mathcal C}^{cl} \setminus {\mathcal D}^{cl}$, and flows again to ${\mathcal C}^{cl}\cap{\mathcal D}^{cl}$ repeating the sequence, the result is a jump instants sequence that is convergent to a finite time value. In this case, the system solution $\phi$ would be a Zeno solution with $\text{dom }\phi = [t_0,t_1]\times\{0\} \cup  [t_1,t_2]\times\{1\} \cup \cdots $, and $0= t_0 < t_1 < \cdots$. But this is not possible since sequences of jumps instants $\{t_i\}_{i=0}^\infty$ can not have subsequences of decreasing jump instants of length greater than $n-n_\rho$ (see \cite{BBbook12}, Prop. 2.4, also \cite{Banos16}). Thus, there always exists a system solution that is unbounded in the $t$-direction.
\end{enumerate}
$\Box$
\end{proof} 

\subsection{Analysis of defective solutions and time-regularized reset control system}
In a reset control system formulated as \eqref{eq:Hcl}, in which the hybrid dynamics is due to the controller (the plant is a LTI continuous-time system), it is important to analyze how hybrid time domains of solutions (see Appendix A) and the non-deterministic behavior of the system are related with the final operation in control practice. For example, for a hybrid time domain that consists of the union of intervals $I^j \times j = [t_j,t_{j+1}]\times j$, with $0 = t_0 < t_1 < t_2 = t_3 = t_4 = t_5 < t_6 <  \cdots$, the solution flows in the time interval  $[t_0,t_1]$, jumps at $t_1$, flows in $[t_1,t_2]$, then it performs three consecutive jumps, keeps flowing in $[t_5,t_6] = [t_2,t_6]$, performs again a jump at $t_6$, etc. From a practical point of view, for solutions to be implemented in a controller, it is necessary to assume that the controller is able to perform a finite number of consecutive jumps instantaneously. In addition, it is compulsory that from any point there always exist solutions that are unbounded in the t-direction. Otherwise, the control system only would present Zeno solutions (genuinely or eventually discrete), that simply can not be implemented in practice, and are considered as defective solutions.

It has been introduced a property, {\em flow persistence}, that is useful to analyze wether a reset control system may be effectively used in control practice regarding the existence of non-defective solutions. Note that if a control system is flow persistent then there always exists a solution which is unbounded in the t-direction; on the contrary, if it is not flow persistent then there may exist points from where all the solutions are bounded in the t-direction, that is there would exist only defective solutions. Thus, although flow persistence is a necessary property in control practice, it is less obvious whether it is a sufficient property, that is, (for a given initial point) is it a problem the existence of Zeno solutions besides solutions that are unbounded in the $t$-direction?. Note that there is always infinite Zeno solutions starting at the points $({\mathbf 0},1), ({\mathbf 0},-1) \in {\mathcal C}^{cl}\cap {\mathcal D}^{cl}$, besides an infinite number of  solutions unbounded in the $t$-direction. 

Another important aspect regarding the final implementation of the hybrid controller is related with its non-deterministic behavior. In principle, the above formulation allows a (finite or infinite) number of different solutions from some initial points. In practice, it is clear that any realistic controller implementation entails a decision such as a solution is selected within all the existing solutions. At this point, a possible answer to the above question is that there is no problem once it is assumed that the controller is able to choose only the solutions that are unbounded in the $t$-direction. However, this type of implementation would require some procedure to properly select the implementable solutions.

 A more simple and common approach to implement the non-deterministic behavior is assume that it is irrelevant the solution that the controller selects, and thus the reset control system would correctly performs for any chosen solution. This is the approach to be followed in this work, and thus it is necessary to remove all the defective solutions. A standard way to avoid the existence of defective solutions is to perform a {\em time regularization} of ${\mathcal H}^{cl}$, introducing a timer $\tau \in [0,\infty)$ that prevent the system to perform two o more consecutive jumps, simply by initializing $\tau$ to $0$ after a jump, and avoiding to perform a new jump until $\tau \geq \tau_m$, where $\tau_m >0$ is a design parameter (usually referred to as the {\em minimum dwell-time}). 

A time-regularized reset control system ${\mathcal H}^{cl}_\tau$ is given by:
\begin{equation} % \small
 {\mathcal H}^{cl}_\tau:
\left \{
\begin{array}{lll}
\dot{\tau} = 1, &
\begin{array}{l}
  \dot{\mathbf{x}}  
\end{array}
=
\begin{array}{l}
  A {\mathbf x} \end{array}  
&\text{, } ({\mathbf x},q,\tau) \in {\mathcal C}^{cl}_\tau
\\
\tau^+ = 0, &
\left(
\begin{array}{l}
  {\mathbf x}^+\\
  {q}^+
\end{array}  \right) =
\left(
\begin{array}{cc}
  A_R& 0\\
  0 &-1
\end{array} \right)
\left(
\begin{array}{l}
   {\mathbf x}\\
   {q}
\end{array} \right)&\text{, } ({\mathbf x},q,\tau) \in {\mathcal D}^{cl}_\tau
\end{array} 
\right.
\label{eq:Hclrho}
\end{equation}
where
\begin{equation}
{\mathcal C}^{cl}_\tau = \mathcal{C}^{cl}\times [0,\infty) \cup \mathcal{D}^{cl}\times [0,\tau_m]
\label{eq:Cclrho}
\end{equation}
and 
\begin{equation}
{\mathcal D}^{cl}_\tau = \mathcal{D}^{cl} \times [\tau_m,\infty). 
\label{eq:Dclrho}
\end{equation}
The following property of of ${\mathcal H}^{cl}_\tau$ easily follows. 

\begin{corollary}
For any $\tau_m >0$, ${\mathcal H}^{cl}_\tau$ is flow persistent and does not have Zeno solutions.
\end{corollary}

\begin{proof}
It trivially follows since for any $(\mathbf{x},q,\tau) \in {\mathcal D}^{cl}_\tau$ it results $(\mathbf{x}^+,q^+,\tau^+)= (A_R\mathbf{x},-q,0)  \in {\mathcal C}^{cl}_\tau \setminus {\mathcal D}^{cl}_\tau$, that is ${\mathcal H}^{cl}_\tau$ always jumps from ${\mathcal D}^{cl}_\tau$ to the interior of ${\mathcal C}^{cl}_\tau$, and then it   flows for at least a time $\tau_m > 0$. 
%For any $\mathbf{x} \in {\mathcal N}({\mathcal O}_\text{base})$, and $\xi = (\mathbf{x},q,\tau) \in {\mathcal C}^{cl} \times [0,\infty)$, since ${\mathcal N}({\mathcal O}_\text{base})$ is $A$-invariant then there always exist a solution $\phi$ with $\phi(t,0) \in {\mathcal N}({\mathcal O}_\text{base})$  for any $t \geq 0$, and $\phi(0,0) = \xi$, which is unbounded in the $t$-direction. If $\mathbf{x} \notin {\mathcal N}({\mathcal O}_\text{base})$ then there exists some $\epsilon > 0$ such as $C \mathbf{x}(t) \neq 0$ for $t \in (0,\epsilon)$, and three cases are possible for solutions starting at $\xi = (\mathbf{x},q,\tau)$, with $q \in \{1,-1\}$: either  $C \mathbf{x}(t) > 0$, or $C \mathbf{x}(t) < 0$ and $\rho \geq \epsilon$, or $C \mathbf{x}(t) < 0$ and $\rho \leq \epsilon$ $\cdots$
$\Box$
\end{proof}

In principle, an election of a small value of the minimum dwell-time $\tau_m$ is all what is needed to prevent the existence of defective solutions. Note, however, that this does not avoid the existence of multiple solutions for some initial conditions, this is for example the case of points $({\mathbf 0},1,0)$ and $({\mathbf 0},-1,0)$. %, and also all the points $(\mathbf{x},q,\tau)$ with $\mathbf{x}$ in the unobservable subspace ${\mathcal N}({\mathcal O}_\text{base})$, simply due to the fact that ${\mathcal N}({\mathcal O}_\text{base})$ is $A$-invariant. 
This non-deterministic behavior will be explored in the next example. 

\begin{example}
Consider the reset control system of Fig. \ref{fig:ICI}, composed by a Horowitz reset controller $R$ and a first-order plant. It will be analyzed its  flow persistence for a exogenous input $w = r$ corresponding to a step reference (no disturbances are present), as well as the influence of $\tau_m$ on the reset control systems solutions.
\begin{figure}[t]
\centerline{\includegraphics[width=11cm,keepaspectratio]{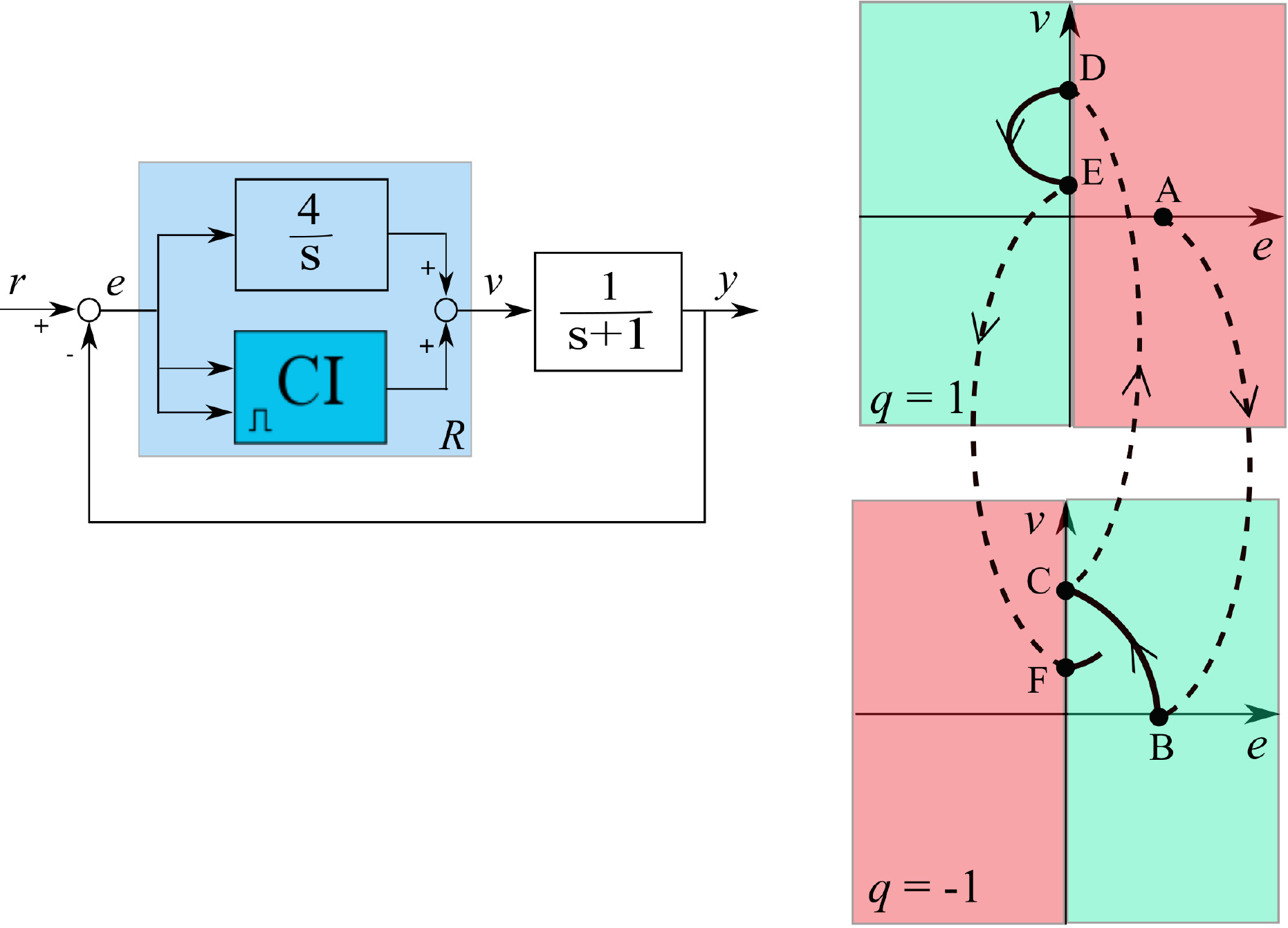}}
\caption{Flow persistence: ({\em left}) Reset control system with a Horowitz reset controller (taken from (\cite{Krishman74}); ({\em right}) Simulation for a unit step reference, with the exosystem, the plant, and the reset controller initially at point  $A = (1,0,0,0,1)$ and $\tau = \tau_m$ (solid lines correspond to flows, an dotted lines to jumps from A to B, from C to D, $\cdots$).} 
\label{fig:ICI}
\end{figure}
For some $\tau_m >0$, the time regularized reset control system is given by \eqref{eq:Hclrho}, with state $(\mathbf{x},q,\tau)$  being $\mathbf{x} = (x_w,x_p,x_{r_1},x_{r_2})$, and
\begin{equation}
\begin{array}{ccc}
 {A}=\left(
\begin{array}[c]{cccc}%
0&0&0&0\\
 0&-1& 1&1\\
 4&-4 &0 &0 \\
1&-1&0&0  %
\end{array}
\right),  &
{A_R}=\left(
\begin{array}[c]{cccc}%
1&0&0&0\\
0&1&0&0\\
0&0&1&0\\
0&0&0&0  %
\end{array}
\right), &
 C=\left(
\begin{array}[c]{cccc}%
1&-1&0&0%
\end{array}
\right), 
 \end{array} \label{eq:AARex}
 \end{equation}
Here the reset controller output is $v = x_{r_1}+ x_{r_2}$ and the error signal is $e = x_w - x_p$. The sets ${\mathcal C}^{cl}$ and ${\mathcal D}^{cl}$ are shown in Fig. \ref{fig:ICI} in the $e$-$v$ planes corresponding to $q =1$ and $q=-1$, as green and red regions, respectively. Note that the flow and jump sets are given by \eqref{eq:Cclrho}-\eqref{eq:Dclrho}, and thus flowing is, in principle, possible in the set ${\mathcal D}^{cl}$. In this example, an initial point $\xi=({\mathbf x}_0,1,\tau_m)$ is considered, and thus the timer allows jumps at the initial intant. 

Consider $\xi$ with  ${\mathbf x}_0 = (1,0,0,0)$.   This corresponds to a unit step reference, and the controller and plant initially at rest. The fact that the reset control system is flow persistent means that there exist a solution that is unbounded in the $t$-direction; this solution is shown in Fig. \ref{fig:ICI} where jumps from A to B, from C to D, $\cdots$ are visible (it is the unique solution for the initial point A).  In this case, the solution is not flowing in the set ${\mathcal D}^{cl}$ since it always flows during a time larger than $\tau_m$ before jumps are enabled. This is true as far as $\tau_m  < \tau_m^\star \approx 1.4416$; in fact, without time regularization, the reset intervals sequence is periodic with fundamental period $\tau_m^\star$ after the second jump. This simple structure of the reset instants sequence is common in low-order reset control systems and in these cases time regularization would not be necessary for most initial points. However, there may exist initial points in which time regularization must be used to remove defective solutions, as it is discussed in the following. 

For an initial point $\xi$ with  ${\mathbf x}_0 = (1,1,1,0)$, corresponding again to a unit step reference but now the controller output is initially $v = 1$ (with the CI initially at rest), it can be easily checked that $\xi \in {\mathcal C}^{cl}_{\tau} \cap {\mathcal D}^{cl}_{\tau}$; moreover, since $A{\mathbf x}_0 = {\mathbf 0}$ and $A_R{\mathbf x}_0 = {\mathbf x}_0$ it directly follows that there exists an infinite number of solutions having one of the following hybrid time domains: $[0,\infty) \times \{0\}$, $[0,t_1)\times \{0\} \cup [t_1,\infty)\times \{1\} $, $[0,t_1]\times \{0\} \cup [t_1,t_2]\times \{1\} \cup [t_2,\infty)\times \{2\}, \cdots$, where $t_1 \in [0,\infty)$ and $t_{j+1} \in [t_j+\tau_m,\infty)$ for $j = 1, 2, \cdots$. That is, there exists an only-flowing solution, and an infinite number of solutions that jumps a finite or infinite number of times. Note that all the solutions are unbounded in the $t$-direction and no Zeno solutions do exist, as far as $\tau_m >0$. Finally, note that any solution $\phi$ satisfies $\phi(t,j) = \xi$ for any $(t,j) \in \text{dom } \phi$; informly speaking, all the solutions produce the same values of controller output and error.

%Note that the base LTI system is not observable, being the unobservable subspace the null space of the observability matrix corresponding to $A$ and $C$, that is 
%\begin{equation}
%{\mathcal N}({\mathcal O}_\text{base}) = \text{span} \{(2,2,1,1),(0,0,-1,1)\}
%\end{equation}
%$\cdots$
\end{example}

%\vspace{1cm}
\subsection{Reset controllers with a sector resetting law}
Although this work is focused on reset control systems with a zero crossing resetting law, is it instructive to analyze other resetting laws that has been developed in the literature. The sector resetting law was introducided in \cite{Zaccarian05}, and has been the main approach within the framework of hybrid inclusions, followed and also extended in several works (\cite{Zaccarian11,Tarbouriech11,Zhao17}, $\cdots$), $\cdots$.

A basic reset controller with a sector resetting law, and with a state $\mathbf{x}_r$ and input $e$, is given by 
\begin{equation}R:
\left \{
\begin{array}{ll}
     \dot{\mathbf{x}}_r  
=
A_r {\mathbf x}_r  
+
  B_r e
 &\text{, } ({\mathbf x}_r,e) \in \mathcal{C} \\
{\mathbf x}_r^+ 
 = A_\rho {\mathbf x}_r
&\text{, } ({\mathbf x}_r,e) \in \mathcal{D}
\end{array} 
\right.
\label{eq:ResetControllerS}
\end{equation}
where ${\mathcal C} = \{ ({\mathbf x_r},e) \in \mathds{R}^{n_r+1}: ev \geq 0 \}$ and  ${\mathcal D} = \{ ({\mathbf x_r},e) \in \mathds{R}^{n_r+1}: ev \leq 0 \}$, being $v = C_r {\mathbf x_r}$ the controller output. Thus, the basic jump set ${\mathcal D}$ is a sector in the $e$-$v$ plane, consisting of its second and fourth quadrants. In combination with the plant \eqref{eq:P} and the exosystems \eqref{eq:ExoS}, and also including time-regularization, the resulting reset control system, with a sector resetting law, is given by
\begin{equation} % \small
% {\mathcal H}^{cl}_\rho:
\left \{
\begin{array}{lll}
\dot{\tau} = 1,&
  \dot{\mathbf{x}}  
=
  A {\mathbf x}   
&\text{, } ({\mathbf x},\tau) \in {\mathcal C}^{cl}_\tau
\\
\tau^+ = 0, &
  {\mathbf x}^+
   = A_R {\mathbf x}
&\text{, } ({\mathbf x},\tau) \in {\mathcal D}^{cl}_\tau
\end{array} 
\right.
\label{eq:HclrhoS}
\end{equation}
where the closed-loop state is now $(\mathbf{x},\tau) = (\mathbf{x}_w,\mathbf{x}_p,\mathbf{x}_r,\tau)$, and the flow and jump sets are also 
$
{\mathcal C}^{cl}_\tau = \mathcal{C}^{cl}\times [0,\infty) \cup \mathcal{D}^{cl}\times [0,\tau_m]
$
and 
$
{\mathcal D}^{cl}_\tau = \mathcal{D}^{cl} \times [\tau_m,\infty) 
$, respectively. But now the sets $\mathcal{C}^{cl}$ and  $\mathcal{D}^{cl}$ defined by the sector resetting law, are given by
\begin{equation}
\mathcal{C}^{cl} = \{ {\mathbf x} \in \mathds{R}^n: {\mathbf x}^T M {\mathbf x} \geq 0 \}
\end{equation}
and 
\begin{equation}
\mathcal{D}^{cl} = \{ {\mathbf x} \in \mathds{R}^n: {\mathbf x}^T M {\mathbf x} \leq 0 \}
\end{equation}
respectively, where 
\begin{equation}
M = 
\left(
\begin{array}{ccc}
O  & O  & C_{w_1}^T C_r   \\
 O & O  &  -C_p^TC_r \\
 C_{r}^T C_{w_1} & -C_r^TC_p   & O   
\end{array}
\right)
\label{eq:M}
\end{equation}
Note that time-regularization may force solutions to flow in the jump set $\mathcal{D}^{cl}$, or in the $e$-$v$ plane to flow in the sector ${\mathcal D}$. It easily follows that this reset control system is flow persistent and that does not have defective solutions.

\begin{figure}[ht]
\centerline{\includegraphics[width=10cm,keepaspectratio]{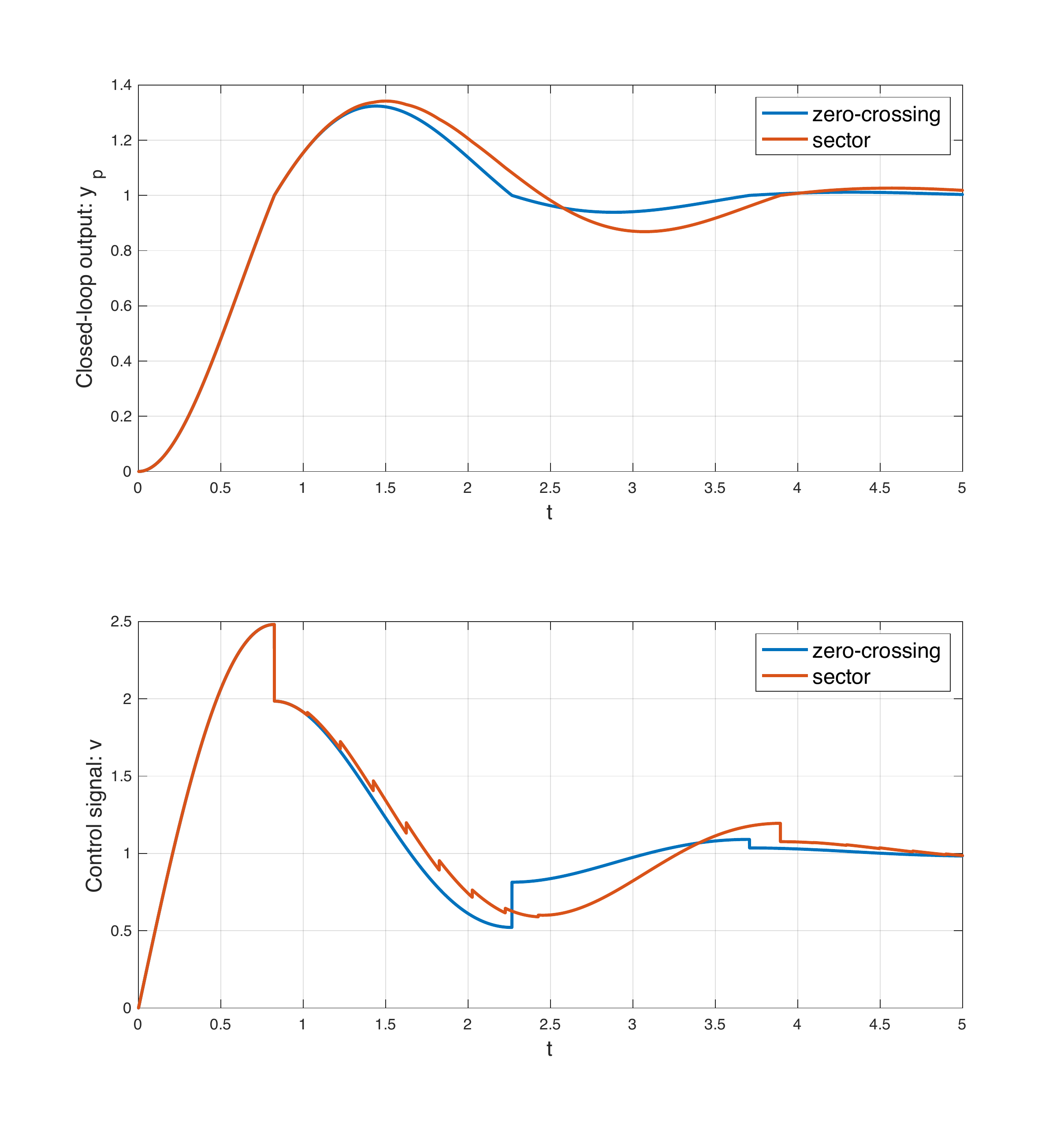}}
\caption{Zero-crossing and sector resetting laws for Example 3.3: ({\em top}) Closed-loop outputs, ({\em botton}) Control signals.}
\label{fig:comparaA}
\end{figure}

\begin{example}
Consider again the reset control system of Fig. \ref{fig:ICI}, with a zero crossing resetting law; and also a reset control system wit a sector resetting law with state $(\mathbf{x},\tau) = (x_w,x_p,x_{r_1},x_{r_2},\tau)$, as given by \eqref{eq:HclrhoS}-\eqref{eq:M}, with $A$ and $A_R$ given by  \eqref{eq:AARex}, and 
\begin{equation}
M = 
\left(
\begin{array}{cccc}
0  & 0  & 1 & 1  \\
0  & 0  & -1&-1 \\
 1 & -1  & 0 & 0 \\
 1 & -1 & 0 & 0 
\end{array}
\right)
\label{eq:MEx}
\end{equation}

Figures \ref{fig:comparaA}-\ref{fig:comparaB} show a simulation of both resetting laws.  Fig. \ref{fig:comparaA} shows the step responses, including closed-loop outputs and control signals. Note that there is an important difference in how both resetting laws performs jumps, specially in the case in which $e < 0$ and $v>0$ and a jump is enabled. In this case, which corresponds for example to the first jump in Fig. \ref{fig:comparaA},  the zero crossing resetting law performs a jump to its flow set and then the solution flows until the next jump at $t \approx 2.25$, while the sector resetting law performs a jump to its jump set. This produces a chattering behavior of the sector resetting law, and in fact the obtained solution would be defective if time-regularization would not have been used. On the other hand, strictly speaking, time-regularization is not necessary for this solution of the zero crossing resetting law, and the same solution is obtained with or without time-regularization (as far as $\tau_m< \tau_m^\star$ -see Example 3.3). The control signals and its component are best analyzed in Fig. \ref{fig:comparaB}: note that after the first jump, in contrast with the zero crossing resetting law, the sector resetting law  periodically reset its $x_{r_2}$ state (with a fundamental period $\tau_m = 0.2$) until $t \approx 2.45$. This is the cause of its chattering, and of its bigger overshoot and undershoot in the step response. The undershoot is specially worst due to the fact that when the error signal changes its sign at $t \approx 2.45$, $x_{r_2} \approx 0$ due to a recent reset.
This example shows that the response of both resetting laws may be very different in general, and although in this case it is clear that the response of the zero-crossing resetting law is qualitatively better in terms of tracking error and control signal chattering, the situation may be different in another cases.  
\end{example}

\begin{figure}[ht]
\centerline{\includegraphics[width=10cm,keepaspectratio]{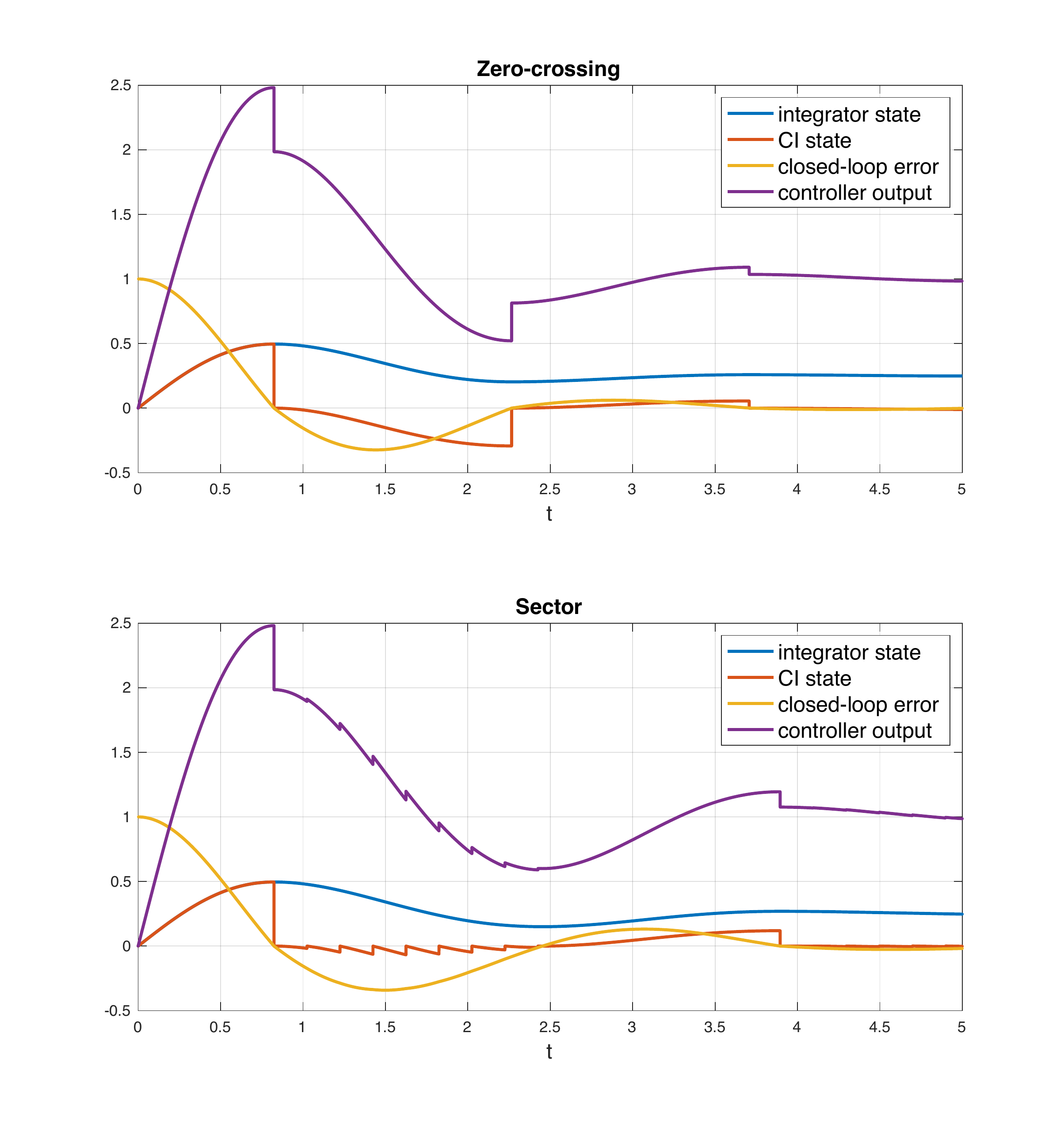}}
\caption{Zero-crossing and sector resetting laws for Example 3.3.}
\label{fig:comparaB}
\end{figure}

\section{Stability analysis}
The stability of the time-regularized reset control system ${\mathcal H}^{cl}_\tau$ with no exogenous inputs %, that is with $ {\bf w} = {\bf 0}$,
is analyzed in the following. Reset times-dependent stability criteria will be developed, inspired from previous results developed in \cite{BBbook12,Banos11}. The basic idea of this approach is to analyze stability by using a discrete-time system (a Poincar\'e-like map),  that represents the sampling of ${\mathcal H}^{cl}_\tau$ at the after-reset instants.

As it is usual in hybrid dynamical systems, stability is referred to sets instead of a single point. The following stability definitions are based on \cite{Goebel09}, note that they are applicable to continuous-time or discrete-time systems as particular cases of hybrid systems. %hat only flow or jump.
Consider a generic hybrid system ${\cal H}$ on ${\mathds R}^n$. For a set ${\cal A} \subset {\mathds R}^n$ and a vector $\phi \in  {\mathds R}^n$, the notation $\| \phi \|_{{\cal A}} = \min \{\| \phi - \psi\|: \psi \in {\cal A}\}$ indicates the distance of $\phi$ to ${\cal A}$. The set  ${\cal A}$ is {\em stable} for ${\mathcal H}$ if for each $\varepsilon >0$ there exist $\delta > 0$ such that $\| \phi(0,0) \|_{\cal A} \leq \delta$ implies $\| \phi(t,j) \|_{\cal A} \leq \varepsilon$ for all solutions $\phi$ to ${\mathcal H}$ and all $(t,j) \in \text{dom }x$. The set ${\cal A}$ is {\em attractive} if there exists a ball ${\mathds B} \subset \mathds{R}^n$ centered at the origin such that 
for any $\xi \in {\mathds B}$ solutions $\phi$ to  ${\mathcal H}^{cl}$ with $\phi(0,0)=\xi$ converge to a set ${\cal A}$, that is  $\| \phi(t,j) \|_{\cal A} \rightarrow 0$ as $t+j \rightarrow \infty$, where $(t,j) \in  \text{dom }\phi$. The set  ${\cal A}$ is {\em asymptotically stable} if it is stable and attractive, and its {\em basin of attraction} is ${\mathds B}$. If the basin of attraction is ${\mathds{R}^n}$ then ${\cal A}$ is {\em globally asymptotically stable}. In the case of the reset control ${\mathcal H}^{cl}_\tau$, stability will be referred to the set ${\cal A}_{\mathbf{0}}=\{\mathbf{0}\}\times \{1,-1\}\times [0,\infty)$. 
\vspace{0.5cm}

For $\xi \in \mathds{R}^n \times \{1,-1\} \times \{0\}$%(where, by simplicity, it is assumed that a reset action is performed at the initial instant and thus the timer is initially at rest)%
, and a solution $\phi = (\mathbf{x},q,\tau)$ to $\mathcal{H}_\tau^\text{cl}$ with $\phi(0,0) = \xi$, and with $\text{dom } \phi = [0,t_1]\times \{0\} \cup [t_1,t_2]\times \{1\} \cup \cdots$, the reset intervals sequence $\tau_\phi$ is defined as ${\tau}_\phi = (\tau_{1},\tau_{2}, \tau_{3}, \cdots)$, where $\tau_j=\tau(t_j,j-1)$ for $j = 1,2, \cdots$, which corresponds to the values of the timer $\tau$ just before jumps are performed. 
Assume, by simplicity, that the timer $\tau$ in ${\mathcal H}^{cl}_\tau$ is initially at rest and that $q(0,0) = 1$, otherwise jumps may always be performed at the initial instant to prepare the system for that state.  Define by simplicity of notation the matrix $J$ as 
\begin{equation}
J =
\left(
\begin{array}{c}
  I_{(n-n_\rho)\times (n-n_\rho)} \\
  O_{n_\rho \times (n-n_\rho)} 
\end{array}
\right)
\end{equation}
such as for any $\mathbf{z} \in \mathds{R}^{n-n_{{\rho}}}$, $J\mathbf{z} =  \left(
\begin{array}{c}
    \mathbf{z}     \\
  \mathbf{0}   
\end{array}
\right) \in \mathds{R}^{n}$. It easily follows that $A_R = JJ^\top$, $J = A_RJ$ and $J^\top = J^\top A_R$. Now, the mapping $I:\mathds{R}^{n-n_{{\rho}}} \rightarrow \mathds{R}_{\geq 0}$, is defined as
\begin{equation}
 I(\mathbf{z}) =  \min \{\tau \geq \tau_m: Ce^{A\tau}
%\left(
%\begin{array}{c}
%    \mathbf{z}     \\
%  \mathbf{0}   
%\end{array}
%\right)
%   \geq 0 \}
J\mathbf{z}    \geq 0 \}
   \label{eq:I}
\end{equation}
Moreover, the following standing assumption will operate in the rest of this work.  It is assumed that there also exist an upper bound of $I$ in the cases in which the flow dynamics is unstable. Note that, otherwise, jumps would be unable to stabilize the flow dynamics. 

%This is a time-varying discrete-time system; it is time-varying due to the fact that the reset interval sequence depends on $xi$ and the particular solution $\phi$. 

\vspace{0.5cm}

\noindent {\em Assumption} $A$: 
\begin{itemize}
\item %A jump is performed at the initial instant, that is $t_1 = 0$ for any solution $\phi$ to  ${\mathcal H}^{cl}_\tau$. 
The reset controller states to be reset and the timer are initially at rest, that is $\mathbf{x}(0,0) = J\mathbf{z}_0$ for some $\mathbf{z}_0 \in \mathds{R}^{n-n_{{\rho}}}$, and $\tau(0,0) = 0$. In addition, $ q(0,0) = 1$.

\item Either the state matrix $A$ in \eqref{eq:Hclrho} is Hurwitz or the mapping $I$ is upper bounded (that is there exists $\Delta_M>0$ such as $I(\mathbf{z}) \leq \Delta_M$ for any $\mathbf{z} \in {\mathds R}^{n- n_{{\rho}}} $).
\end{itemize}

Here, a a Poincar\'e-{\em like} map, that gives the evolution from one after-jump state to the next one with a sign change, is postulated.
By definition, the discrete-time system ${\cal DH}_\tau^{cl}$, with state $\mathbf{z} \in \mathds{R}^{n_p+n_{\hat{\rho}}}$,  is  given by 
%${\mathbf z}^+ = g(\mathbf{z})$, where
\begin{equation}
{\mathbf z}^+ = g(\mathbf{z}) = - J^T e^{A \cdot I(\mathbf{z})}  
    J\mathbf{z}    
\label{eq:PM}
\end{equation}
Note that sign change allows, starting at $q=1$, to obtain the successive reset intervals by using the mapping $I$, wich does not explicitly depends on $q$. Also to obtain a simplified dynamic discrete system ${\cal DH}_\tau^{cl}$ in which the state $\mathbf{z}$ only consists of the first $n-n_\rho$ values of $\mathbf{x} = J \mathbf{z}$, and that will be perfectly valid to analyze the stability of the original hybrid control system as it will be seen in the following.  Some homogeneity properties of the maps $I$ and $g$ easily follow: for any $\lambda >0$ it is true that
\begin{equation}
\begin{array}{ll}
(i) & I(\lambda \mathbf{z}) = I(\mathbf{z})    \\
 (ii) & g(\lambda \mathbf{z}) = \lambda  g(\mathbf{z})     \\
  &     
\end{array}
\label{eq:propHom}
\end{equation}

%$\cdots$ FALTAN DEFINICIONES PRECISAS DE ESTABILIDAD. 

\begin{proposition}  The set ${\cal A}_{\mathbf{0}}$ is (globally) asymptotically stable for the reset control system ${\mathcal H}^{cl}_\tau$ %(with basin of attraction ${\mathds B}_d \times \mathds{R}^{n_\rho} \times \{1,-1\} \times [0,\Delta_M] $)
 if and only if the origin $\{\mathbf{0}\}$ is (globally) asymptotically stable for the discrete-time system ${\cal DH}_\tau^{cl}$. %(with basin of attraction ${\mathds B}_d$).
\end{proposition}

\begin{proof} It is an adaptation of \cite{Banos11}-Prop. 3.1 to the hybrid formalism adopted in this work. According to Assumption A, and from an initial state $({\mathbf x}_0,q_0,\tau_0) \in \mathds{R}^n \times \{1,-1\} \times \mathds{R}$, the reset control system ${\mathcal H}^{cl}_\tau$ is prepared (forcing jumps if necessary) to be in a state  $\xi = (J{\mathbf z}_0,1,0)=(({\mathbf z}_0,{\mathbf 0}),1,0)$, with ${\mathbf z}_0 \in {\mathds R}^{n_p+n_{{\hat \rho}}}$, that it is redefined to be the initial state. Also, consider a solution $\phi = ({\mathbf x},q,\tau) $ to ${\mathcal H}^{cl}_\tau$ with $\phi(0,0) =\xi$. % \in \mathds{R}^n \times \{1,-1\}\times \{0\}$,  and a solution $\phi = ({\mathbf x},q,\tau) $ to ${\mathcal H}^{cl}_\tau$ with $\phi(0,0) =\xi$. From its definition it is clear that $\| \xi\|_{\mathcal A} = \|{\mathbf x}_0\|$ and $\| \phi\|_{\mathcal A} = \|{\mathbf x}\|$.
 
\vspace{0.25cm}
({\em only if}) From the definition of ${\mathcal H}^{cl}_\tau$ and its jump set, it follows that the values of $\phi$ at the after-reset instants (including the initial value) are given by 
\begin{equation}
\begin{array}{l}
 \phi(0,0) = (J\mathbf{z}_0,1,0)\\
 \phi(t_1,1) = (-J\mathbf{z}(1),-1,0) \\
 \phi(t_2,2) = (J\mathbf{z}(2),1,0)    \\
 \cdots   
\end{array}
\label{eq:secphi}
\end{equation}
%the sequence of after-jump states is $\{((\mathbf{z}_0, \mathbf{0}),1,0),((-\mathbf{z}_1, \mathbf{0}),-1,0),((\mathbf{z}_2, \mathbf{0}),1,0) , \cdots \}$
%It follows directly from the fact that $\phi(t_j,j) = \phi_d(j)$ for any $(t_j,j) \in \text{dom} \phi $.
If $\{ \mathbf{ 0}\}$ is not stable for ${\cal DH}_\tau^{cl}$ then there exists an $\varepsilon >0$ such that $\|\phi(t_j,j)\|_{\mathcal{A}_{\mathbf{0}}} =  \| \mathbf{z}(j) \| > \varepsilon$ for some $(t_j,j) \in  \text{dom} \phi$. %, for a solution $\phi$ to ${\mathcal H}^{cl}_\tau$ with $\phi(0,0) = \xi$ and $\xi = (\phi_d(0),0)$.
As a result, ${\mathcal A}_{\mathbf{0}}$ is not stable for ${\mathcal H}^{cl}_\tau$. On the other hand, if $\{ \mathbf{ 0}\}$  is not atractive for ${\cal DH}_\tau^{cl}$  then there will exist a sequence of values $\| \phi(t_j,j) \|_{\mathcal{A}_{\mathbf{0}}}  = \| \mathbf{z}(j) \|$, for $j = 1, 2, \cdots$ which does not converge to zero, and thus ${\mathcal A} $ will not be atractive for ${\mathcal H}^{cl}_\tau$.

%any $\xi \in \mathds{R}^n \times \{1,-1\}\times \{0\}$

\vspace{0.25cm}
({\em if}) From the flow equation in \eqref{eq:Hclrho} it follows that for $(t,j) \in \text{dom }\phi$ 
\begin{equation}
{\mathbf x}(t,j) = {\mathbf x}(t_j,j) + \int_{t_j}^t A{\mathbf x}(s,j)ds 
\end{equation}
and applying the Gronwall inequality (observing that the induced norm $\|A\|$ is always bounded by some real number $\alpha$ and that $\|J\mathbf{z} \| = \|\mathbf{z}\|$) it directly follows that 
\begin{equation}
\|{\mathbf x}(t,j) \| \leq \|{\mathbf x}(t_j,j)\|e^{\alpha(t-t_j)} = \|(-1)^{j+1}J{\mathbf z}(j)\|e^{\alpha(t-t_j)} = \| \mathbf{z}(j)\|e^{\alpha(t-t_j)} 
\end{equation}
and thus
\begin{equation}
\|{\phi}(t,j) \|_{\mathcal{A}_{\mathbf{0}}}  \leq \|\mathbf{z}(j)\|e^{\alpha(t-t_j)} 
\label{eq:xnorm}
\end{equation}

 Since $\{\mathbf{0}\}$ is stable for ${\cal DH}_\tau^{cl}$ it is true that there exists $\gamma >0$ such as $\|\mathbf{z}(j)\| \leq \gamma \|\mathbf{z}_0\| = \gamma \|\xi\|_{\mathcal{A}_{\mathbf{0}}} $ and thus \eqref{eq:xnorm} results in 
\begin{equation}
\|{\phi}(t,j) \|_{\mathcal{A}_{\mathbf{0}}}\leq \gamma e^{\alpha(t-t_j)} \|\xi\|_{\mathcal{A}_{\mathbf{0}}}
\end{equation}
Now, from assumption $A$ it follows that either $\alpha < 0$ or $t-t_j < \Delta_M$ and thus stability of the set ${\mathcal{A}_{\mathbf{0}}}$ for ${\mathcal H}^{cl}_\tau$ directly follows. Asymptotic stability is obtained from the fact that for the discrete-time system ${\cal DH}_\tau^{cl}$ it is true that $\| \mathbf{z}(j)\| < \gamma \lambda^j\| \mathbf{z}_0\|$, for $j = 1,2, \cdots$, and for any  $ \mathbf{z}_0 \in {\mathds B}$ and $0 \leq \lambda < 1$ (here ${\mathds B} \subset {\mathds R}^{n_p+n_{{\hat \rho}}}$ is a ball centered at the origin). Substituting in \eqref{eq:xnorm} and using again the fact that $ \|\xi\|_{\mathcal{A}_{\mathbf{0}}} =\|\mathbf{z}_0\| $,  it results that
\begin{equation}
\|{\phi}(t,j) \|_{\mathcal{A}_{\mathbf{0}}} \leq \gamma \lambda^je^{\alpha(t-t_j)} \|\xi\|_{\mathcal{A}_{\mathbf{0}}} 
\end{equation}
and thus, since either $\alpha <0$ or $t-t_j < \Delta_M$, it directly follows that $\| \phi(t,j) \|_{\mathcal{A}_{\mathbf{0}}}  \rightarrow 0$ as $t+j \rightarrow \infty$, where $(t,j) \in  \text{dom }\phi$, and $\xi \in ({\mathds B} \times \{ \mathbf{0} \})\times \{1\} \times \{0\}$. %, which will be the basin of attraction of ${\mathcal A}$ .
Finally, since the initial condition $\xi$ is prepared from any $({\mathbf x}_0,q_0,\tau_0) \in \mathds{R}^n \times \{1,-1\} \times \mathds{R}$, global asymptotic stability of ${\mathcal{A}_{\mathbf{0}}}$ for ${\mathcal H}^{cl}_\tau$ comes from the global asymptotic stability of $\{\mathbf{0}\}$ for ${\cal DH}_\tau^{cl}$.
$\Box$
\end{proof}

\vspace{0.25cm}

%\begin{corollary}
%\end{corollary}

\subsection{Stability based on periods of the reset interval sequences}
%In the following, stability of ${\cal P}$, which is equivalent to the stability of the reset control system ${\mathcal H}^{cl}_\tau$ according to Prop. 4.1, will be analyzed for some particular cases that are relevant in control practice, in which reset interval sequences have a particularly simple structure.  
 %An important previous question posed in \cite{BBbook12} is related with the structure of the reset intervals sequences. 
It is well known  that in some relevant cases in control practice, reset intervals sequences 
have a particularly simple structure (\cite{BBbook12}).  For example, for $n_r = n_\rho$ (full reset) and a second-order plant, that is $n_p =2$, reset intervals sequences are periodic with a constant fundamental period $\Delta$, after the second reset (this is also the case of Example 3.3 with no exogenous inputs as far as $\tau_m$ is small enough). Note that in these cases, stability may be simply checked by analyzing if $A_R e^{A \Delta}$ is a Schur matrix. 

In the following, it will be investigated how periodic patterns related with the dicrete time system ${\cal DH}_\tau^{cl}$ may be used to obtain stability criteria for the reset control system ${\mathcal H}^{cl}_\tau$. 
%are eventually periodic, that is are periodic removing some finite subsequence at the begining, and its period is a constant $\Delta$ independently of the solution. This is for example the case n found (\cite{BBbook12}) other periodic patterns that may simplify the stability analysis, including cases of asymptotically periodic interval sequences. %In this Section, it will be used the following standing assumption.
%Assume by simplicity that $q_0 = 1$ and that a reset is performed at the initial instant. 
Regarding ${\cal DH}_\tau^{cl}$, consider the equivalence relation $\equiv$ : for $\mathbf{z}_1$,$\mathbf{z}_2 \in \mathds{R}^{n-n_\rho}\setminus{\mathbf{0}}$, $\mathbf{z}_1 \equiv \mathbf{z}_2$ if there exists $\lambda >0$ such as  $\mathbf{z}_1= \lambda \mathbf{z}_2$. Each equivalence class can be represented by a unit vector $\mathbf{s}$ in the unit $(n-n_\rho-1)$-sphere ${S}^{n-n_\rho-1}$. Now, it is defined the {\em angle mapping} %s $I_{\pm 1}:{\mathds S}^{n-n_\rho-1}\rightarrow \mathds{R}$, such as $I_{\pm 1}({\mathbf z})  = I(\left[
%\begin{array}{c}
%{\small {\mathbf z}}   \\
%  {\mathbf 0}
%\end{array}
%\right]
%, \pm 1)$, and the map 
%$G:{\mathds S}^{n-n_\rho-1} \rightarrow {\mathds S}^{n-n_\rho-1}$ such as 
%\begin{equation}
%{\mathbf z}^+ = G({\mathbf z}) = - \frac
%{
%A_Re^{A \cdot I_1({\mathbf z})}
%\left[
%\begin{array}{c}
%{\small {\mathbf z}}   \\
%  {\mathbf 0}
%\end{array}
%\right]
%}
%{\|A_Re^{A \cdot I_1({\mathbf z})}\left[
%\begin{array}{c}
%{\small {\mathbf z}}   \\
%  {\mathbf 0}
%\end{array}
%\right]
%\|}
%\label{eq:G}
%\end{equation}
$\varPi_g:{S}^{n-n_\rho-1} \rightarrow {S}^{n-n_\rho-1}$ as
\begin{equation}
{\mathbf s}^+ = \varPi_g({\mathbf s}) =  \frac
{
g({\mathbf s})}
{\|g({\mathbf s})\|}
\label{eq:PIg}
\end{equation}

Note that $\mathbf{s}$ can be interpreted as the projection of $\mathbf{z} \in \mathds{R}^{n-n_\rho}$ on %the unit $(n-n_\rho-1)$-sphere 
${S}^{n-n_\rho-1} $, and thus the mapping $\varPi_g$ will produce orbits of those projections. A natural form of analyzing periodic interval sequences of the reset control system ${\mathcal H}^{cl}_\tau$ is by analyzing periodic points of $\varPi_g$. These points will define the periodic structure of the reset intervals sequences allowing to develop an asymptotic stability criterion. 

Some definitions about stability of periodic points follows (see, for example, \cite{Alligood00} and \cite {Luo12} for technical details). Firstly, $\varPi_g^k(\mathbf{s})$ is defined to be the result of applying $k$ times $\varPi_g$ to the point $\mathbf{s}$. The {\em orbit} of $\mathbf{s}$ under $\varPi_g$ is the set of points $\{\mathbf{s} =\varPi_g^0(\mathbf{s}), \varPi_g(\mathbf{s}), \varPi_g^2(\mathbf{s}), \cdots, \}$. ${\mathbf p}$ is a {\em periodic-k point} if $\varPi_g^k({\mathbf p}) = {\mathbf p}$ and if $k$ is the smaller such positive integer; and the orbit of ${\mathbf p}$ with $k$ points, that is $\{\mathbf{p}, \varPi_g(\mathbf{p}), \cdots, \varPi_g^k(\mathbf{p})\}$, is called a {\em periodic-k orbit}. For $k = 1$, ${\mathbf p}$ is referred to as a {\em fixed point}. Assume that $\varPi_g$ is differentiable in a neighborhood $U$ of a fixed point ${\mathbf p}$ and let ${\mathbf D}\varPi_g({\mathbf p})$ be the Jacobian matrix of $G$ at ${\mathbf p}$; the fixed point ${\mathbf p}$ is called a {\em sink} if ${\mathbf D}\varPi_g({\mathbf p})$ is a Schur matrix, and a {\em source} is all eigenvalues of ${\mathbf D}\varPi_g({\mathbf p})$ has a magnitude greater than 1. The stable manifold of ${\mathbf p}$, denoted as ${\cal S}({\mathbf p})$, is the set of points ${\mathbf s} \in {S}^{n-n_\rho-1} $ such that $\|\varPi_g^k(\mathbf{s}) - {\mathbf p} \| \rightarrow 0$ as $k \rightarrow \infty$. Analogously, for a periodic-$k$ point ${\mathbf p}$, its periodic-$k$ orbit is a {\em sink} ({\em source}) if ${\mathbf p}$ is a sink (source) for the map $\varPi_g^k$.

%A fixed point ${\mathbf p}$ is a {\em sink}  or an {\em attracting fixed point} if there exist $\varepsilon >0$ such that for all ${\mathbf z} \in \varepsilon {\mathds B}({\mathbf p})$, $G^n({\mathbf z}) \rightarrow {\mathbf p}$ as $n \rightarrow \infty$. Analogously, for a periodic-$k$ point ${\mathbf p}$, its periodic-$k$ orbit is a {\em sink} if ${\mathbf p}$ is a sink for the map $G^k$.

%For a reset control system ${\mathcal H}^{cl}_\rho$ define the map $z^+ = G(z)$.  

\begin{proposition} 
Assume that the angle map $\varPi_g$ has a periodic-k point ${\mathbf p}$, being $\varPi_g^k$ differentiable in a neighborhood $U$ of $\mathbf{p}$, and that its periodic-$k$ orbit is a sink with stable manifold $ {\mathcal S}(\mathbf{p})={S}^{n-n_\rho-1} $. 
Then, ${\mathbf p}$ is an eigenvector of the matrix 
 \begin{equation}
	M_{\mathbf{p}} := J^\top e^{A \cdot I(\varPi_g^{k-1}({\mathbf p}))}\cdots A_Re^{A\cdot I(\varPi_g({\mathbf p}))}A_Re^{A\cdot I({\mathbf p})} J
\label{eq:Mp}
\end{equation}
corresponding to a real eigenvalue $\lambda_{\mathbf p}$, and the 
set ${\mathcal A}_{\mathbf{0}} $ is asymptotically stable for the reset control system ${\mathcal H}^{cl}_\tau$ if and only if
$|\lambda_{\mathbf p}| < 1$. Moreover, the basin of atraction of ${\mathcal A}_{\mathbf{0}}$ is $\mathds{R}^n\times\{1,-1\} \times [0,\infty)$ (stability is global).

% . Moreover, the basin of atraction of ${\mathcal A}_{\mathbf{0}}$ contains the set 
%\begin{equation}
%\mathds{B}_{{\mathcal A}_{\mathbf{0}}} (\mathbf{p})= \{ (\lambda J\mathbf{s},q,\tau) \in \mathds{R}^{n}\times\{1,-1\}\times [0,\infty) :  %\mathbf{s} \in  {\mathcal S}(\mathbf{p}), \lambda >0\}
%\end{equation}

%with a reset interval sequence $\Delta_\phi$ which is a stable periodic-k orbit of  $\Delta^+={\cal I}(\Delta)$, is asymptotically stable if and only $A_Re^{A f^k(\Delta)}A_Re^{A f^{k-1}(\Delta)} \cdots A_Re^{Af(\Delta)}A_Re^{A\Delta}$ is Hurwitz. 

\end{proposition}

\begin{proof}
%Consider a point $\mathbf{z}_0 \in \mathds{R}^{n - n_{{\rho}}} \setminus \{\mathbf{0}\}$, and 
%Consider a a solution $\mathbf{z}$ to ${\cal DH}_\tau^{cl}$ with $\mathbf{z}(k) \in \mathds{R}^{n - n_{{\rho}}} \setminus \{\mathbf{0}\}$, for $k= 0,1,2, \cdots$. Otherwise,  if $\mathbf{z}(k) = 0$ for $k \geq k^*$ and some finite $k^*$ then ${\cal DH}_\tau^{cl}$ would trivially be globally asymptotically stable. 
Consider the discrete system ${\cal DH}_\tau^{cl}$ as given by \eqref{eq:PM}, $\mathbf{z}_0 \in  \mathds{R}^{n-n_\rho}$, $\mathbf{z}_k = g^k(\mathbf{z}_0)$, and also $\mathbf{s}_0 =  \frac{\mathbf{z}_0}{\|\mathbf{z}_0\|}$, and $\mathbf{s}_k =  \varPi_g^k({\mathbf s}_0)$,  for $k = 1,2,\cdots$. From the homogeneity property \eqref{eq:propHom}  and \eqref{eq:PIg} it directly follows that 
 %\begin{equation}
%\{ I(\mathbf{z}(0)) , I(\mathbf{z}(1)) , I(\mathbf{z}(2)) , \cdots \} =  \{ I(\frac{\mathbf{z}(0)}{\|\mathbf{z}(0)\|}), I(\frac{\mathbf{z}(1)}{\|\mathbf{z}(1)\|}), I(\frac{\mathbf{z}(2)}{\|\mathbf{z}(2)\|}) , \cdots \}, 
%\end{equation}
\begin{equation}
 I(\mathbf{z}_k) =   I(\mathbf{s}_k) 
 \label{eq:Izs} 
\end{equation}
%and defining $\mathbf{s}_0 =  \frac{\mathbf{z}(0)}{\|\mathbf{z}(0)\|}$, and $\mathbf{s}_k =  \varPi_g^k({\mathbf s})$ for $k= 1,2, \cdots$, 
%from \eqref{eq:PIg} it results that 
%\begin{equation}
% \{ I(\mathbf{z}(0)) , I(\mathbf{z}(1)) , I(\mathbf{z}(2)) , \cdots \} = \{ I(\mathbf{s}_0) , I(\mathbf{s}_1) , I(\mathbf{s}_2) , \cdots \}
%\label{eq:sectau1}
%\end{equation}
for $k = 0,1,2, \cdots$. The proof is particularized for the case in which  ${\mathbf p}$ is a fixed point of the angle map $\varPi_g$ (that is $k = 1$); for the cases $k = 2, 3, \cdots$, the proof is similar, using $\varPi_g^k$ instead of $\varPi_g$ in the following reasoning. Now, define the matrix functions $M, \delta M:\mathds{R}^{n-n_\rho} \rightarrow  \mathbb R^{(n-n_\rho ) \times (n-n_\rho )}$, such as $M(\mathbf z) = J^\top e^{AI(\mathbf{z})}J$, and thus $M(\mathbf p)  = M_{\mathbf p}$ as given by \eqref{eq:Mp}, and $\delta M({\mathbf z}) = M(\mathbf z)- M_{\mathbf p}$. 

Firstly, it will be shown that ${\mathbf p}$ is an eigenvector of $M_{\mathbf p}$. Since $\mathbf p$ is a sink with stable manifold
the whole sphere, then it is true that 
$\|{\mathbf s}_k - {\mathbf p} \| \rightarrow 0$ as $k \rightarrow \infty$, for any $\mathbf s_0 \in \mathds{S}^{n-n_\rho}$. Moreover, since the mapping $I$ is continuous 
at $\mathbf p$ (otherwise $\Pi_g$ would not be differentiable), and thus the mapping $M$ is also continuous at $\mathbf{p}$, it also follows that $\|I({\mathbf s}_k) - I({\mathbf p}) \| \rightarrow 0$ and finally 
\begin{equation}
\|M({\mathbf s}_k) - M_{\mathbf p} \| \rightarrow 0 
\label{eq:MJsp} 
\end{equation}
%, as $k \longrightarrow \infty$:
as $k \rightarrow \infty$, for any $\mathbf s_0 \in \mathds{S}^{n-n_\rho}$. Then, by directly using \eqref{eq:PM} and \eqref{eq:PIg}, two points ${\mathbf s}_k$ and ${\mathbf s}_{k+1} = \varPi_g({\mathbf s}_{k})$ are related by 
\begin{equation}
{\mathbf s}_{k+1} = -\frac{M({\mathbf s}_k){\mathbf s}_k}{\|M({\mathbf s}_k){\mathbf s}_k\|}
\end{equation}
or equivalently
\begin{equation}
M({\mathbf s}_k){\mathbf s}_k = -\|M({\mathbf s}_k){\mathbf s}_k\|{\mathbf s}_{k+1}
%{\mathbf s}_{k+1} = \frac{M({\mathbf s}_k){\mathbf s}_k}{\|M({\mathbf s}_k){\mathbf s}_k\|}
\end{equation}
and finally for $k \rightarrow \infty$ it results that 
\begin{equation}
M_{\mathbf p}{\mathbf p} = -\|M_{\mathbf p}{\mathbf p}\|{\mathbf p}
%{\mathbf s}_{k+1} = \frac{M({\mathbf s}_k){\mathbf s}_k}{\|M({\mathbf s}_k){\mathbf s}_k\|}
\end{equation}
that is ${\mathbf p}$ is an eigenvector of $M_{\mathbf p}$ with eigenvalue $\lambda_{\mathbf p} := -\|M_{\mathbf p}{\mathbf p}\|$, a (non-positive) real number.
 
Secondly, consider the (unique) orthogonal decomposition of ${\mathbf z}_k$ as
\begin{equation}
{\mathbf z}_k  =  \alpha_k {\mathbf p} + \alpha_k \delta_k{\mathbf p}^\perp
\end{equation}
where ${\mathbf p}^\perp$ is a vector perpendicular to ${\mathbf p}$, and  $\alpha_k$  and $\delta_k$ are real numbers, for any $k = 1, 2, \cdots$.  Since it is true that $\|{\mathbf s}_k - {\mathbf p} \| \rightarrow 0$ as $k \rightarrow \infty$, then it directly follows that $\delta_k \rightarrow 0 $ as $k \rightarrow \infty$. In addition, consider two points ${\mathbf z}_k$ and ${\mathbf z}_{k+1} = g({\mathbf z}_k)$, for $k = 1, 2, \cdots$. It results that
\begin{equation}
\frac{\|{\mathbf z}_{k+1}\|}{\|{\mathbf z}_{k}\|} = \frac{\|-M({\mathbf z}_{k}){\mathbf z}_{k}\|}{\|{\mathbf z}_{k}\|} = 
 \frac{\|(M_{\mathbf p} + \delta M_k)(\alpha_k {\mathbf p} + \alpha_k \delta_k{\mathbf p}^\perp )\|}{\|\alpha_k {\mathbf p} + \alpha_k \delta_k{\mathbf p}^\perp\|} \rightarrow \frac{\|M_{\mathbf p} {\mathbf p}\| }{\|{\mathbf p} \|} = |\lambda_{\mathbf p} |
\label{eq:zplusz}
\end{equation}
for $k \rightarrow \infty$, where by definition $\delta M_k = \delta M({\mathbf z}_k)$, and it has been used the fact that $\|\delta M_k \| \rightarrow 0$ as $k \rightarrow \infty$ (it easily follows from \eqref{eq:MJsp}).\\

({\em if}) If $|\lambda_{\mathbf p}| < 1$ then from \eqref{eq:zplusz} if follows that for some real number $\lambda$ such as $0 \leq |\lambda_{\mathbf p}| < \lambda < 1$, there exist an integer $N$ such as $\|{\mathbf z}_{k+1}\| = \|M({\mathbf z}_{k}){\mathbf z}_{k}\| < \lambda \|{\mathbf z}_{k}\| $, for $k \geq N$, and this is true for any point ${\mathbf z}_0$. Since, in addition, $\|M({\mathbf z})\| \leq e^{\|A\|\Delta_M}$ it results that
\begin{equation} \small
\|{\mathbf z}_{k}\| \leq \|M({\mathbf z}_{k-1})\| \|M({\mathbf z}_{k-2})\| \cdots \|M({\mathbf z}_{N})\|\|M({\mathbf z}_{N-1})\| \cdots \|M({\mathbf z}_{0})\|\|{\mathbf z}_{0}\| \leq \lambda^{k-N}(e^{\|A\|\Delta_M})^{N}\|{\mathbf z}_{0}\| 
\end{equation}
that is 
\begin{equation} \small
\|{\mathbf z}_{k}\| \leq \gamma \lambda^{k}\|{\mathbf z}_{0}\| 
\end{equation}
for $k \geq N$, where $\gamma = (e^{\|A\|\Delta_M}/\lambda)^{N}$ is a constant. It is also true that $\|{\mathbf z}_{k}\| \leq (e^{\|A\|\Delta_M})^{k}\|{\mathbf z}_{0}\|$, for $k < N$. It easily follows that the origin is globally asymptotically stable for ${\cal DH}_\tau^{cl}$ and Prop. 4.1 certifies global asymptotic stability of  ${\mathcal A}_{\mathbf{0}} $ for the reset control system ${\mathcal H}^{cl}_\tau$.

\vspace{0.25cm}
({\em only if}) Consider an initial condition ${\mathbf z}_{0}\ = \alpha_0 {\mathbf p}$, where $\alpha_0$ is a non-zero arbitrary number. On the other hand, $M(c{\mathbf p}) = M({\mathbf p})$ for any constant $c$. Thus, ${\mathbf z}_1 = - M({\mathbf z}_{0}){\mathbf z}_{0} = - M(\alpha_0 {\mathbf p}) \alpha_0 {\mathbf p} = - \alpha_0 M({\mathbf p}) {\mathbf p} = - \alpha_0 \lambda_{\mathbf p} {\mathbf p}$, and it is obtained that ${\mathbf z}_k =  (-1)^k \alpha_0 \lambda_{\mathbf p}^k {\mathbf p}$ for $k = 1, 2, \cdots$. Stability of ${\mathcal A}_{\mathbf{0}} $ for ${\mathcal H}^{cl}_\tau$ implies stability of the origin for ${\cal DH}_\tau^{cl}$, and this implies that $|\lambda_{\mathbf p}| < 1$.
$\Box$

\end{proof}

Note that, in general, the angle map $\varPi_g$ may exhibit several periodic-$k$ points, not necessarily sinks, with different integer values of $k\geq1$, each periodic point having its own stable manifold. For example, in the case in which the point is a source the stable manifold is the point itself. The following result follows using similar arguments to the proof of Prop. 4.2.  

\begin{corollary}
Assume that the angle map $\varPi_g$ has a finite number of periodic points $\mathbf{p}_1,\mathbf{p}_2, \cdots,\mathbf{p}_n$, with periods $k_1, k_2, \cdots,k_n$, respectively, and that
\begin{equation}
\bigcup_{i=1}^n \mathcal{S}(\mathbf{p}_i) = {S}^{n-n_\rho-1}. 
\end{equation}
Then, for each $i = 1, 2,\cdots, n$, the point ${\mathbf p}_i$ is an eigenvector of the matrix   
\begin{equation}
M_{\mathbf{p}_i} = A_Re^{A \cdot I(\varPi_g^{k_i-1}({\mathbf p}_i))}\cdots A_Re^{A\cdot I(\varPi_g({\mathbf p}_i))}A_Re^{A\cdot I({\mathbf p}_i)}
\end{equation}
corresponding to a real eigenvalue $\lambda_{{\mathbf p}_i}$, 
and the set ${\mathcal A}_{\mathbf{0}} $ is asymptotically stable for the reset control system ${\mathcal H}^{cl}_\tau$ if and only if
$|\lambda_{{\mathbf p}_i}|<1$ for any $i = 1, 2, \cdots,n$. Moreover, the basin of atraction of ${\mathcal A}_{\mathbf{0}}$ is $\mathds{R}^n\times\{1,-1\} \times [0,\infty)$ (stability is global).
\end{corollary}

\subsection{A case study with the Horowitz reset controller}
For the Example 3.3 without exogenous inputs, the reset control system ${\mathcal H}^{cl}_\tau$ is given by the matrices  $A$, $A_R$ and $C$ obtained by removing the first column and the first row of the matrices in \eqref{eq:AARex}, that is   
\begin{equation}
\begin{array}{ccc}
 {A}=\left(
\begin{array}[c]{ccc}%
-1& 1&1\\
 -4 &0 &0 \\
-1&0&0  %
\end{array}
\right),  &
{A_R}=\left(
\begin{array}[c]{ccc}%
1&0&0\\
0&1&0\\
0&0&0  %
\end{array}
\right), &
 C=\left(
\begin{array}[c]{ccc}%
-1&0&0%
\end{array}
\right), 
 \end{array} \label{eq:AARex2}
 \end{equation}
Here $A$ has eigenvalues $\lambda_{1,2} = -\frac{1}{2} \pm j\frac{\sqrt{19}}{2}$ and $\lambda_{3} = 0$. 
%The mode corresponding to $\lambda_3= 0$ is $\mathbf{v}_3 = (0,1,-1)$, and since $C\mathbf{v}_3 = 0$ then it is unobservable. This means that for any  $\xi = ({\mathbf x}_0,q_0)
%\in \text{span} \{\mathbf{v}_3\} \times \{1,-1\}$, there exists a solution $\phi(t,0) = \xi$ for any $t \geq 0$ corresponding to a zero closed-loop output ($y = 0$), due to the fact that $(A,C,\mathbf{v}_3 )$ is not reduced for $\lambda_{1,2}$. And, in addition, there exists infinite solutions for which a uniform upper bound of the first crossing interval does not exist. 
After preparing ${\mathcal H}^{cl}_\tau$ according to Assumption A, %it is true that $(A,C,\mathbf{x}_0 )$ is  reduced for $\lambda_{1,2}$ and thus
an oscillatory error signal is always produced; moreover,  it is not difficult to obtain that $\Delta_M = \frac{2\pi}{\sqrt{19}}+ \tau_m$ (it is the value of the minimum dwell-time added to half oscillation of the error signal). %= \rho\star \approx 1.4416$  is an upper bound of the reset intervals, as far as $\rho < \rho^\star$. 
In this case $n-n_\rho-1 =1$, and thus $\varPi_g$ maps values of ${\mathbf s}$ in the unit circle ${S}^1$.  Moreover, by using the parametrization ${\mathbf s} = [\cos(\theta),\sin(\theta)]$, for $\theta \in [-\pi,\pi]$, it is possible to redefine the angle map (with some abuse of notation) as $\varPi_g:[-\pi,\pi] \rightarrow [-\pi,\pi]$, and $\theta^+ = \varPi_g(\theta)$. %It results that, as it may be expected, the minimum dwell-time constant $\tau_m$ has a major impact in the analysis of the periodic orbits of $\varPi_g$. 
Fig. \ref{fig:puntosfijosHor} shows the graph of $\varPi_g$ for two values of $\tau_m$. 

For $\tau_m = 0.25$, $\varPi_g$ has a jump discontinuity at the point $\theta = d$; this point corresponds to the value of $\theta$ such that $Ce^{A\tau_m}[\cos(\theta),\sin(\theta),0] = 0$, and it can be obtained is a closed-form after some computation, it is given by $d = \pi + \frac{1}{2}(1-\frac{{\sqrt 19}}{\tan (\frac{{\sqrt 19}}{8})})$. Moreover,  $\varPi_g$ has a unique fixed point at $\theta = p = \pi/2$, which is a sink with a basin of attraction $\mathcal{S}(\pi/2) = [-\pi,\pi]$. The direct application of Prop. 4.2 results in that the reset control system ${\mathcal H}^{cl}_{\tau}$ is globally asymptotically stable, since $I(p) = I(\pi/2) =  2 \pi / {\sqrt 19}$ and 
\begin{equation}
M_p = J^\top e^{A\cdot I(p)}J = 
\left(
\begin{array}{cc}
-0.4863  & 0      \\
  0& -0.1891    
 \end{array}
\right)
\label{eq:S025}
\end{equation}
has an eigenvector $(1,0)$ (corresponding to $p = \pi/2$) with eigenvalue $-0.4863$.

%$\cdots$ LOS PUNTOS PERIODICOS DEPENDEN BASTANTE DE RO, COMO ES ESPERABLE.
%PARA VALORES PEQUENOS HAY UN PUNTO 1-PERIODICO QUE RESULTA EN UN $\Delta = 2\pi/\sqrt{19}$ (ESPERABLE). PARA VALORES MAS GRANDES CREO QUE SE PUEDEN OBTENER PUNTOS 1-PERIODICOS E INCLUSO 2-PERIODICOS CON DIFERENTES VALORES DE DE DELTA

\begin{figure}[t]
\centering
{\includegraphics[width=1\textwidth]{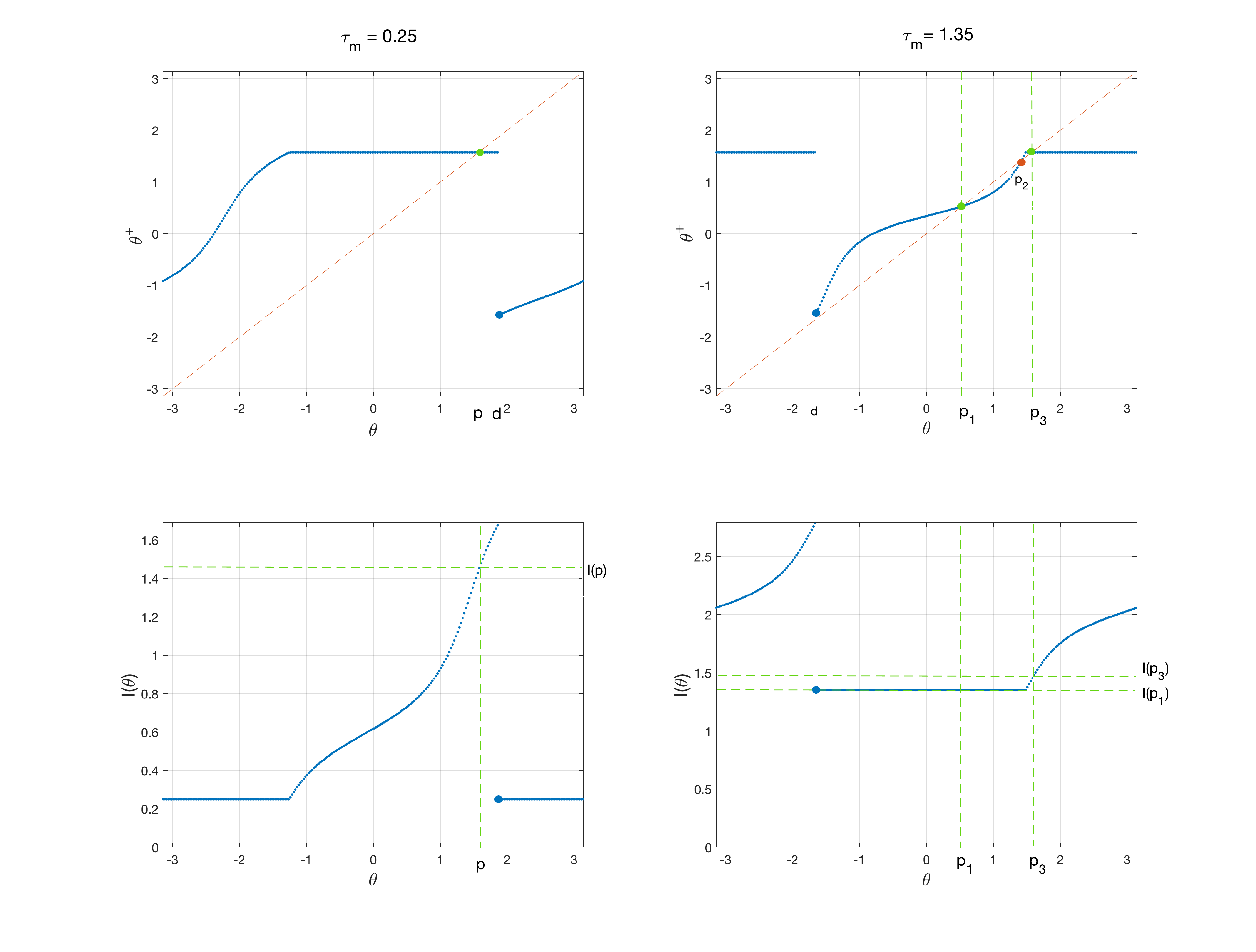}}
\caption{ Graphs of the maps $\varPi_g$ ({\em top}) and $I$ ({\em bottom}) for the reset control system of Example 3.3:  ({\em left}) $\tau_m = 0.25$, ({\em right}) $\tau_m = 1.35$- }
\label{fig:puntosfijosHor}
\end{figure}

For $\tau_m = 1.35$, the mapping $\varPi_g$ exhibits a more complex structure (Fig. \ref{fig:puntosfijosHor}). Besides a jump discontinuity at $\theta = d = - \pi + \frac{1}{2}(1-\frac{{\sqrt 19}}{\tan (\frac{{\sqrt 19}}{2}\cdot1.35)})$, $\varPi_g$ has three fixed points: $p_1 \approx 0.5322$, $p_2  \approx 1.4246$, and $p_3 = \pi/2$. Both $p_1$ and $p_3$ are sinks, while $p_2$ is a source. Their basins of attraction  are ${\mathcal S}(p_1)= [d,p_2)$, ${\mathcal S}(p_2) = \{p_2\}$ and ${\mathcal S}(p_3)= [-\pi,d) \cup (p_2,\pi]$. Finally, $I(p_1) = I(p_2) = \tau_m = 1.35$ and $I(p_3) = 2\pi/\sqrt{19}$. As a result, since ${\mathcal S}(p_1)\cup{\mathcal S}(p_2)\cup{\mathcal S}(p_3)= [-\pi,\pi]$, and the three matrices
\begin{equation}
M_{p_1} = J^\top e^{A \cdot I(p_1)} J= 
\left(
\begin{array}{ccc}
-0.5222  & 0.0463     \\
  -0.1850& -0.1808     \\
 \end{array}
\right),
\end{equation}
$M_{p_2} = M_{p_1}$, and $M_{p_3} = A_Re^{A \cdot I(p_3)} = M_p$ (as given in \eqref{eq:S025}) are Schur matrices (and thus all eigenvalues are strictly inside the unit circle), it also follows that ${\mathcal H}^{cl}_{\tau}$ is globally asymptotically stable for $\tau_m = 1.35$.

It turns out that the fixed points and periodic point patterns are heavily influenced by the minimum dwell-time $\tau_m$ as it should be expected. Although, in principle, in control practice $\tau_m$ is initially used to avoid defective solutions and a small value of it is all what is needed (and following the above analysis is not difficult to show that  ${\mathcal H}^{cl}_\tau$ is globally asymptotically stable for any $\tau_m < 1.35$), it is illustrative to analyze how the periodic point patterns change with it. Although an exhaustive analysis is out of scope of this work and it will be given elsewhere, there are several bifurcation points delimiting zones with one sink, with two sinks plus a source, with periodic-2 sinks,  with periodic-3 sinks, etc.  To conclude the example, a case with periodic-3 sinks is considered in the following. 

%= G(p_3) = G(p_1) G(p_2) 

For $\tau_m = 2$, there are three periodic-3 points (Fig. \ref{fig:puntosfijosHor2}). They are $p_1  \approx -1.5494$, $p_2 \approx -0.2162$, and $p_3  = \pi/2$. %, and the union of the three basins of attraction is $[-\pi,\pi]$. 
For $p_1$, its periodic-3 orbit  $\{p_1,\varPi_g(p_1),\varPi_g^2(p_1)\} = \{p_1,p_3,p_2\}$ is a sink and, in addition, $I(p_1) \approx 2.9042$, and $  I(p_3) = I(p_2) =\tau_m = 2$, and
\begin{equation}
M_{p_1} = J^\top e^{A \cdot I(p_2)} A_Re^{A \cdot I(p_3)} A_Re^{A \cdot I(p_1)} J = 10^{-2} \cdot
\left(
\begin{array}{ccc}
-2.2711  & 0.0337    \\
  0.0011& -3.8597    \\
\end{array}
\right),
\end{equation}
is a Schur matrix. For $p_2$ and $p_3$,  its periodic-3 orbits  are $\{p_2, p_1,p_3\}$ and $\{p_3, p_2,p_1\}$, respectively. They are also sinks, and the matrices $M_{p_2} = J^\top e^{A \cdot I(p_3)} A_Re^{A \cdot I(p_1)} A_Re^{A \cdot I(p_2)} J$ and $M(p_3) = J^\top e^{A \cdot I(p_1)} A_Re^{A \cdot I(p_2)} A_Re^{A \cdot I(p_3)}J$ can be easily checked to be Schur matrices. Finally, applying Corollary 4.3,  using the fact that the union of the three basins of attraction is $[-\pi,\pi]$, it results that ${\mathcal H}^{cl}_{\tau}$ is globally asymptotically stable for $\tau_m = 2$.

\begin{figure}[t]
\centering
{\includegraphics[width=\textwidth]{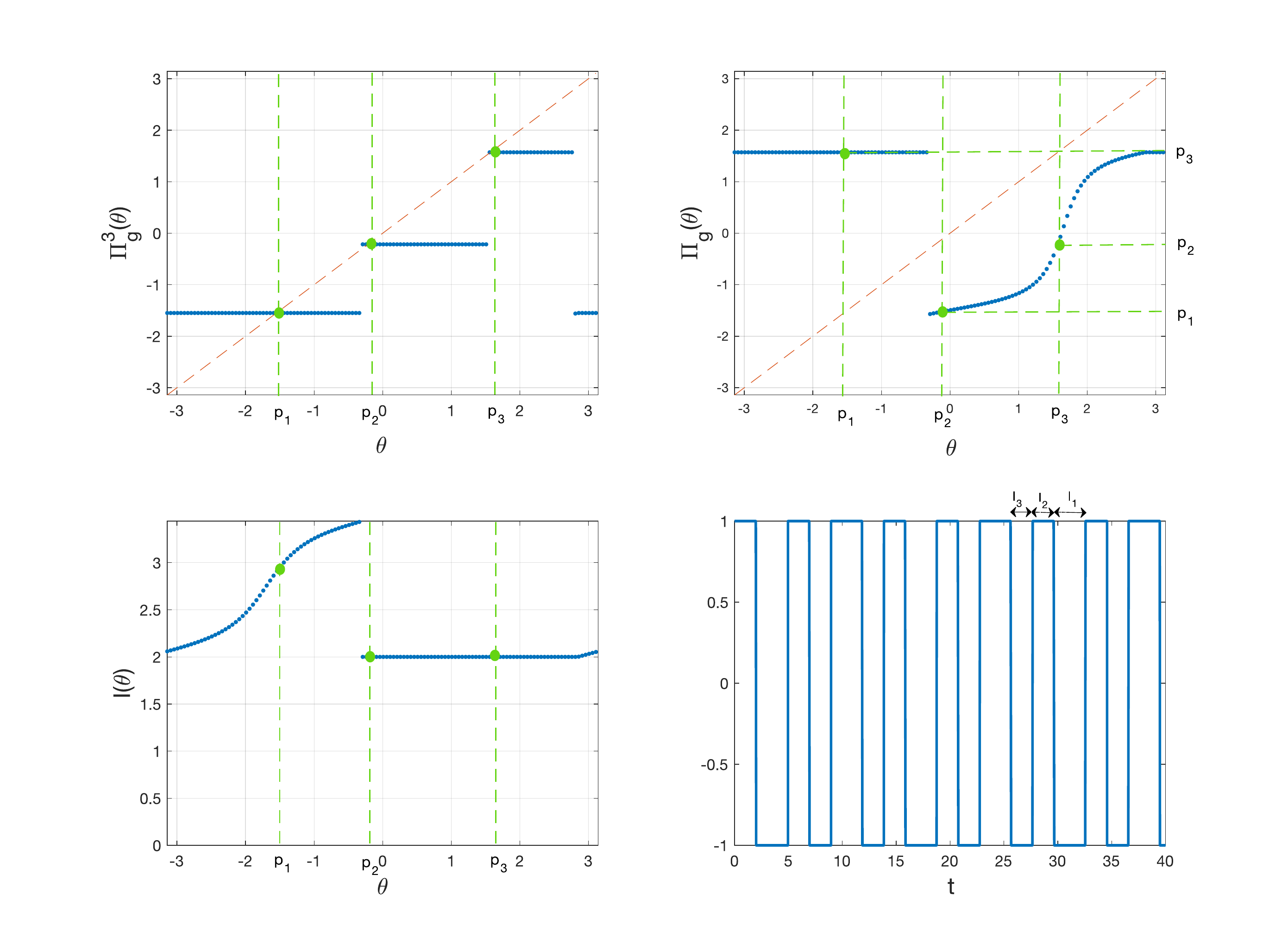}}
\caption{Reset control system of Example 3.3 for $\tau_m = 2$: ({\em top}) Graph of the mapping $\varPi_g^3$ ({\em left}), showing the periodic-3 points $p_1$, $p_2$, and $p_3$,  and graph of $\varPi_g$ ({\em right}); ({\em bottom}) graph of the map $I$ ({\em left}), and plot showing the periodicity of the reset intervals ({\em rigth}), where $I_3 = I(p_3) = 2$, $I_2 = I(p_2 ) = 2$, and $I_1 = I(p_1)\approx 2.9042$.}
\label{fig:puntosfijosHor2}
\end{figure}

%\end{example}

%\subsection{Another stability criterion $\cdots$}
\subsection{A case with a chaotic sequence of reset intervals}
Obviously,  Prop. 4.2-Corol. 4.3 are convenient in practice for those cases in which periodic orbits of $\varPi_g$ can be found with a reasonable effort (for a given value of $\tau_m$), like in the cases analyzed above. Although, in general, the result is useful in many practical cases,  even some low-order reset control systems may exhibit extraordinarily complex interval patterns that makes elusive its application in those cases, motivating the investigation of alternative stability criteria. In the following, it is analyzed a reset control system consisting of a FORE and a third order plant, that produces chaotic sequences of reset intervals. 
%\begin{example}{(An example with a chaotic sequence of reset intervals)}
Consider the time-regularized reset control system $\mathcal{H}_{\tau}^\text{cl}$, as given by \eqref{eq:Hclrho}, with no exogenous inputs and with $\tau_m = 0.1$, and with:
\begin{equation}
\begin{array}{ccc}
 {A}=\left(
\begin{array}[c]{cccc}%
0&0&3.5&5\\
 1&0& -4.3&1\\
 0&1 &-1 &0 \\
0&0&-1&-1  %
\end{array}
\right),  &
{A_R}=\left(
\begin{array}[c]{cccc}%
1&0&0&0\\
0&1&0&0\\
0&0&1&0\\
0&0&0&0  %
\end{array}
\right), &
 C=\left(
\begin{array}[c]{cccc}%
0&0&-1&0%
\end{array}
\right), 
 \end{array} \label{eq:AARcaos}
 \end{equation}
This reset control system, with state $(\mathbf{x},q,\tau) = (\mathbf{x}_p,x_r,q,\tau)$, is the feedback composition of a FORE with state $(x_r,q)$ and a third order plant with state $\mathbf{x}_p = (x_1,x_2,x_3)$. Here, the mapping $\varPi_g: {S}^2 \rightarrow  {S}^2  $ has an invariant set ${\cal T} = \{ (\frac{t}{\sqrt{1+t^2}}, \frac{1}{\sqrt{1+t^2}},0)  \in {S}^2: t \in [-3,4]\}$, that is $\varPi_g({\cal T}) \subset {\cal T}$.  And thus, $\varPi_g$ can be parameterized (with some abuse of notation) as $\varPi_g:[-3,4] \rightarrow [-3,4]$, where $t^+ = \varPi_g(t)$ (see Fig. \ref{fig:caos}). 
\begin{figure}[t]
\centerline{\includegraphics[width=\textwidth]{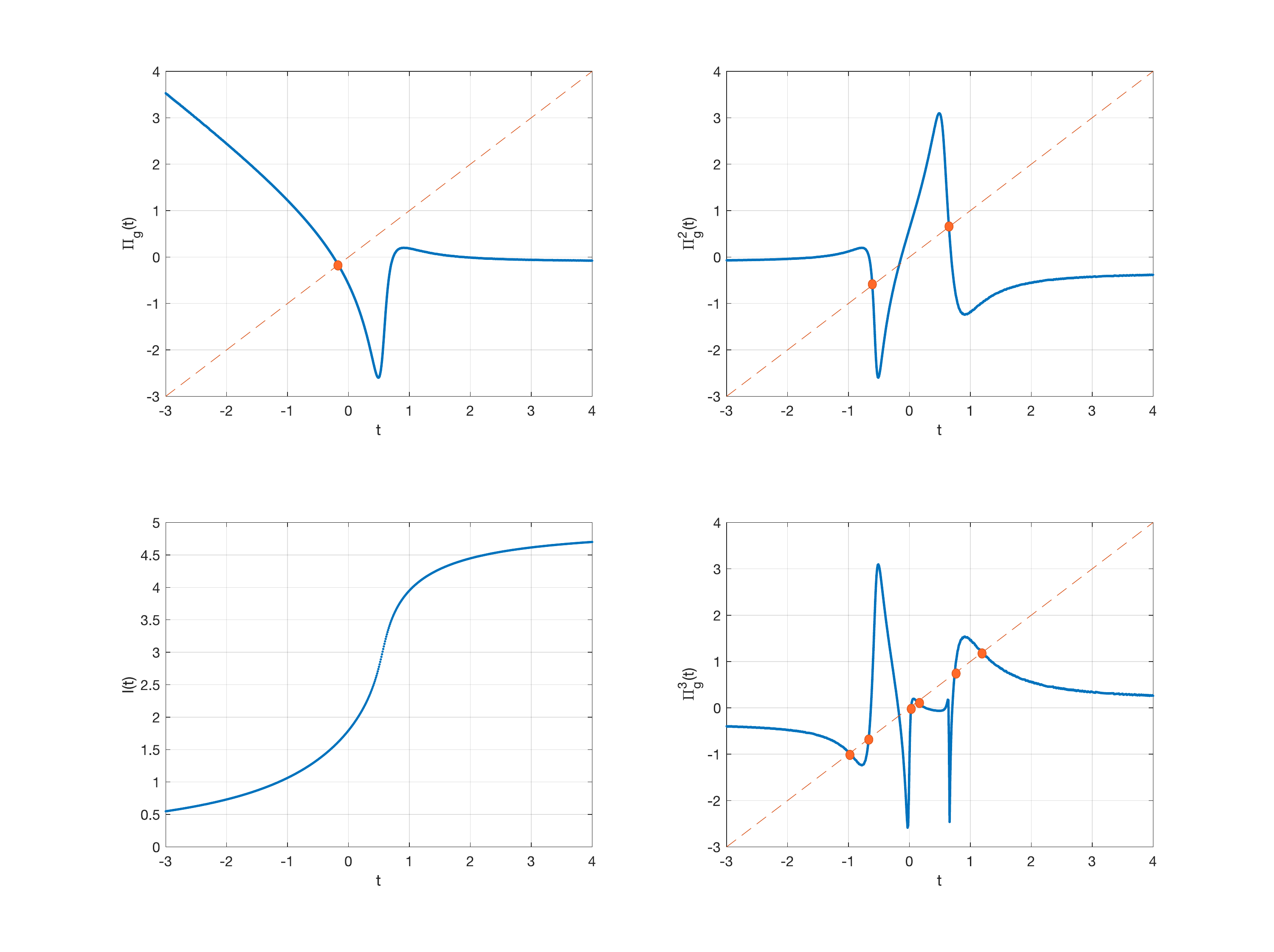}}
\caption{An example of reset control system with a chaotic sequence of reset intervals: graphs of the mapps $\varPi_g$, $\varPi_g^2$, and $\varPi^3_g$ (with marks at the periodic points), and $I$.}
\label{fig:caos}
\end{figure}

%\end{example}

For this case, $\varPi_g$ is a continuous mapping on $[-3,4]$, its graph is shown in Fig. \ref{fig:caos}. %and \eqref{eq:caos} can be interpreted as a discrete dynamic system (Poincar\' e map) that represent the transition between the two after-reset values $x$ and $x^+$. 
Moreover the mapping $I:[-3,4] \rightarrow [0.1,\infty)$, such as $I(t)$ (also with some abuse of notation) defines the reset interval from the point $\mathbf{s} = (\frac{t}{\sqrt{1+t^2}}, \frac{1}{\sqrt{1+t^2}},0)$. Since the mapping $\varPi_g$ is continuous on the interval $[-3,4]$ and has periodic-$3$ points (see Fig \ref{fig:caos}), then it turns out that there exists solutions with any possible integer period according to Sharkovskii theorem (\cite{Shark64},\cite{Alligood00}), that is there exists periodic-$k$ points for $k= 1, 2, 3, \cdots$. For example, Fig. \ref{fig:caos} shows the graphs of $\varPi_g$, $\varPi_g^2$, and $\varPi_g^3$, explicitly marking 1 periodic-1 point (a fixed point), 2 periodic-2 points, and 6 periodic-3 points. Note that all the periodic points are sources, and in fact this is the case for any periodic point since, as it is well known, period 3 implies chaos \cite{LiYorke75}. As a result, any initial point of $\mathcal{H}_{\tau}^\text{cl}$, with $\xi \in (\frac{t}{\sqrt{1+t^2}}, \frac{1}{\sqrt{1+t^2}},0,0) \times \{1\} \times \{0\}$, and $t\in [-3,4]$, will produce solutions with chaotic sequences of reset intervals, making elusive to apply Prop. 4.2.

\subsection{Reset-times dependent stability conditions: Minimum dwell-time}
An alternative approach to stability analysis of ${\mathcal H}^{cl}_{\tau}$ will be based on the use of Lyapunov functions, with an explicit consideration of the reset interval sequences $\tau_\phi = \{\tau_1, \tau_2, \cdots\}$ corresponding to solutions $\phi = ({\mathbf x},q,\tau)$ to ${\mathcal H}^{cl}_{\tau}$. This approach is based on the reset-times dependent stability criteria early developed in \cite{Banos11,BBbook12}. The set of all possible reset interval sequences is defined as 
\begin{equation}
S_{{\mathcal H}^{cl}_{\tau}}= \{\tau_\phi = \{\tau_1, \tau_2, \cdots \} \subset {\mathds R}_{\geq 0}: \tau_i = \tau(t_i,i), (t_i,i) \in \text{ dom } \phi, \phi \text{ is a solution to } {\mathcal H}^{cl}_{\tau} \}
\label{eq:SHtau}
\end{equation}

In the case in which $A$ is a Hurwitz matrix, a first strategy consists in embedding the set of reset intervals sequences in a larger set characterized by the minimum dwell-time associated to ${\mathcal H}^{cl}_{\tau}$. It is defined the set $S_{\tau_m}$ as
\begin{equation}
S_{\tau_m} = \{\{\tau_1, \tau_2, \cdots\} \subset {\mathds R}_{\geq 0}: \tau_i \geq \tau_m \}
\end{equation}
and then stability conditions are considered for any posible reset interval sequence in $S_{\tau_m}$. This approach will allow a direct application of computationally efficient methods imported from the impulsive systems literature.  And although, in principle, results may be  conservative due to the fact that $S_{{\mathcal H}^{cl}_{\tau}}$ is a meager set compared to $S_{\tau_m} \supset S_{{\mathcal H}^{cl}_{\tau}}$, in practice they allow to obtain a first value of $\tau_m$ for which stability is guaranteed.%in the case in which $A$ is a Hurwitz matrix this problem is alleviated as it will be seen in the following.  

\begin{proposition}
%Consider the reset control system ${\mathcal H}^{cl}_{\tau_m}$ with a Hurwitz matrix $A$. 
The set ${\cal A}_0$ is globally asymptotically stable for 
${\mathcal H}^{cl}_{\tau}$ if there exist a sequence of positive definite matrices $\{P_1,P_2, \cdots\}$, such that
\begin{equation}
\begin{array}{l}
  \eta I \leq P_k \leq \rho I    \\
  \\
  e^{A^T \tau_k} A_R P_{k+1} A_R e^{A \tau_k} - P_k \leq -\varepsilon I
\end{array}
\label{eq:prop43a}
\end{equation}
hold for $k = 1, 2, \cdots$, for some positive constants $\eta$, $\rho$, and $\varepsilon$, and any $\{\tau_1, \tau_2, \cdots\} \in S_{\tau_m}$.
\end{proposition}

\begin{proof}
Firstly, it is a standard result (\cite{Rugh96}) that if \eqref{eq:prop43a} hold then the time-dependent quadratic Lyapunov function $V(\mathbf{x},k)=\mathbf{x}^TP_k \mathbf{z}$ certifies that every discrete-time (time-varying) system 

\begin{equation}
\left \{
\begin{array}{l}
 \mathbf{z}^+  
 = -J^T A_Re^{A\tau_k}
J \mathbf{z}  \\
%\tau^+ = \tau_{k+1}     \\

k^+ = k+1    
\end{array}
\right.
\label{eq:prop43b}
\end{equation}
%\begin{equation}
%\left \{
%\begin{array}{l}
%\left (
%\begin{array}{l}
% \mathbf{z}^+  \\
% \mathbf{0}
%\end{array}
%\right)
%= -A_Re^{A\tau_k}
%\left (
%\begin{array}{l}
% \mathbf{z}  \\
% \mathbf{0}
%\end{array}
%\right ) \\
%%\tau^+ = \tau_{k+1}     \\
%\\
%k^+ = k+1    
%\end{array}
%\right.
%\label{eq:prop43b}
%\end{equation}
with $\{\tau_1, \tau_2, \cdots\} \in S_{\tau_m}$, is globally asymptotically stable.

%e $S_{{\mathcal H}^{cl}_{\tau}} \subset S_{\tau_m} $, and,
%In addition, for any $\tau_\phi \in S_{{\mathcal H}^{cl}_{\tau}}$ with $\phi(0,0) = ((\mathbf{z}_0,\mathbf{0}),1,0)$ it is true that $\tau_\phi= \{I(\mathbf{z}_0), I(g(\mathbf{z}_0)), I((g^2(\mathbf{z}_0)), \cdots\}$, and thus 
Now, since it is true that $I(\mathbf{z}) > \tau_m$ for any $\mathbf{z} \in \mathds{R}^{n_p+n_{\hat{\rho}}}$ (it directly follows from \eqref{eq:I}), and thus the solution $\mathbf{z}$ to  ${\cal DH}_\tau^{cl}$ with $\mathbf{z}(0) = \mathbf{z}_0$ corresponds with the solution to a discrete-time system like \eqref{eq:prop43b} with $\{\tau_1, \tau_2, \tau_3, \cdots\}= \{I(\mathbf{z}_0), I(g(\mathbf{z}_0)), I((g^2(\mathbf{z}_0)),\cdots \} \in S_{\tau_m}$, then it follows from \eqref{eq:prop43a} that  (see \cite{Rugh96}, Th. 23.3)
\begin{equation}
\|\mathbf{z}(k)\|^2 = \|g^k(\mathbf{z}_0)\|^2   \leq \frac{\rho}{\eta}\lambda^{2k}\|\mathbf{z}_0\|^2
\end{equation}
for $k \geq 0$, where $\lambda  <1$. %$\lambda = \sqrt{1-\nu/\rho} <1$. 
Since the constants $\rho$, $\eta$, and $\lambda$ do not depend on $\mathbf{z}_0$ then it directly follows that the origin $\{\mathbf{0}\}$ is globally asymptotically stable for the discrete time system ${\cal DH}_\tau^{cl}$. Application of Prop. 4.1 ends the proof.
$\Box$ \\

\end{proof}

%\vspace{3 cm}
%
%This means that for each $\epsilon$ there exist
%
%In the following, based on that fact, it will be shown that 
% $\{ \mathbf{0}\}$ is globally asymptotically stable for ${\cal DH}_\tau^{cl}$. In fact, it will be shown that for a regularization of it, $%\widehat{{\cal DH}}_\tau^{cl}$, given by 
%\mathbf{z}^+ \in G(\mathbf{z}) = \bigcap_{\delta >0} \overline{g(\mathbf{z}+\delta\mathds{B})}
%\end{equation}
%the origin will be globally asymptotically stable, which is a stronger property (\cite{Kellett04}). %Consider the quadratic Lyapunov function %$V:\mathds{R}^{n_r+n_{\hat{\rho}}} \rightarrow \mathds{R}_{\geq 0}$, given by $V(\mathbf{z}) = \mathbf{z}^TP \mathbf{z}$
%Note that the mapping $g$ has at most jump discontinuities due to the fact that the mapping $I$ only has jump discontinuities. These %discontinuity at the points $\mathbf{z}$ such that 

%
%
%\vspace{3cm}

%Since $G$ is upper-semicontinuous, using Th. 20 in \cite{Goebel09}  

Prop. 4.4 gives a nice and simple connection between stability of the reset control system ${\mathcal H}^{cl}_{\tau}$ and stability of impulsive systems with impulses at fixed instants, since condition \eqref{eq:prop43a} can be easily linked with the stability of an impulsive system with impulses at instants $t_k = t_{k-1} + \tau_k$, $k = 1, 2, \cdots$. Moreover, the following Corollary directly follows.

\begin{corollary} Assume that \eqref{eq:prop43a} hold for $k = 1, 2, \cdots$, for  $\eta$, $\rho$, $\varepsilon >0$, and for any $\{\tau_1, \tau_2, \cdots\} \in S_{\tau_m^\star}$. Then the set ${\cal A}_0$ is globally asymptotically stable for 
${\mathcal H}^{cl}_{\tau}$, for any $\tau_m \geq \tau_m^*$.

\end{corollary}

A particular simple instance of \eqref{eq:prop43a} is to consider a time independent Lyapunov function, that is $P_k = P$, for any $k = 1, 2, \cdots$. In this case, the procedure for solving \eqref{eq:prop43a} is reduced to searching for a matrix $P>0$ such as 
\begin{equation}
  e^{A^T \tau} A_R P A_R e^{A \tau} - P \leq -\varepsilon I
\label{eq:prop43c}
\end{equation}
for some $\varepsilon > 0$ and any $\tau \geq \tau_m$ (or $\tau \geq \tau_m^*$)
This problem is well-known in the literature and there exists several good methods for its resolution (\cite{libroResetBB,Banos11,Dash2013,Briat2013}. The next result is directly imported from \cite{Briat2013}. 
%This allows a more simple procedure in searching simply a matrix $P>0$ such that 

\begin{corollary}
Assume that one of the following conditions applies:
\begin{itemize}
%Consider the reset control system ${\cal H}_\tau^{cl}$, given by \eqref{eq:Hclrho}.
\item There exist a differentiable matrix function $R:[0,\tau_m^*] \rightarrow \mathds{S}^n$, $R(0) > 0$, and $\varepsilon > 0$ such that

%The following conditions are equivalent:
%\begin{itemize}
%\item There exist a matrix $P > 0$ such that
%\begin{equation}
%A^TP + PA < 0
%\end{equation}
%and
%\begin{equation}
%e^{A^T \tau_m}A_RPA_Re^{A \tau_m} - P < 0
%\end{equation}
%hold.

%\item  There exist a differentiable matrix function $R:[0,\tau_M] \rightarrow \mathds{S}^n$, $R(0) > 0$, and $\varepsilon > 0$ such that
\begin{equation}
\begin{array}{l}
 A^TR(0) + R(0)A < 0 \\
 A^TR(\theta) + R(\theta)A+ {\dot R}(\theta) \leq 0, \\
 A_R R(0) A_R - R(\tau_m^*) \leq -\varepsilon I, 
\end{array}
\label{eq:cor1a}
\end{equation}
hold for any $\theta \in [0,\tau_m^*]$.
%\end{itemize}
%Moreover, if one of the above conditions holds, the set ${\cal A}_{\mathbf{0}}$ is globally asymptotically stable for ${\cal H}_\tau^{cl}$.

%Consider the reset control system ${\cal H}_\tau^{cl}$, given by \eqref{eq:Hclrho}.
\item There exist a differentiable matrix function $S:[0,\tau_m^*] \rightarrow \mathds{S}^n$, $S(0) > 0$, and $\varepsilon > 0$ such that

%The following conditions are equivalent:
%\begin{itemize}
%\item There exist a matrix $P > 0$ such that
%\begin{equation}
%A^TP + PA < 0
%\end{equation}
%and
%\begin{equation}
%e^{A^T \tau_m}A_RPA_Re^{A \tau_m} - P < 0
%\end{equation}
%hold.

%\item  There exist a differentiable matrix function $R:[0,\tau_M] \rightarrow \mathds{S}^n$, $R(0) > 0$, and $\varepsilon > 0$ such that
\begin{equation}
\begin{array}{l}
A^TS(\tau_m^*) + S(\tau_m^*)A < 0 \\
A^TS(\theta) + S(\theta)A - {\dot S}(\theta) \leq 0, \\
A_R S(\tau_m^*) A_R - S(0) \leq -\varepsilon I, 
\end{array}
\label{eq:cor1b}
\end{equation}
hold for any $\theta \in [0,\tau_m^*]$.
\end{itemize}
then the set ${\cal A}_{\mathbf{0}}$ is globally asymptotically stable for ${\cal H}_\tau^{cl}$, for any $\tau_m \geq \tau_m^*$.
\end{corollary}

\begin{proof}
It is a direct application of Th. 2.3 in \cite{Briat2013}, and Prop. 4.3-Corollary 4.4, for the case given by \eqref{eq:prop43c}. $\Box$
\end{proof}

This Corollary results in an efficient procedure for solving the ${\cal H}_\tau^{cl}$ stability problem, specifically when the matrix function $R$ (or $S$) is searched in the set of matrix polynomials with a given degree $d_R$, that is $R(\theta) = \sum_{i=0}^{d_R}R_i\theta^i$, $R_i >0$. In this case, \eqref{eq:cor1a} and \eqref{eq:cor1b} may be inserted in sum-of-squares conditions, and the problem can be solved by using some sum-of-squares programming package (see e.g. \cite{Praj2004}). %In the following example, several cases that follows this procedure are reported. 

\begin{example}
This is a classical example of reset control system, considered in early works about reset control like \cite{Beker04}, It consists of a feedback interconnection of a FORE and a second order plant. Here, it is defined as a time-regularized reset control system ${\cal H}_\tau^{cl}$ with 
\begin{equation}
\begin{array}{ccc}
A = 
\left(
\begin{array}{ccc}
0  & 0  & 1  \\
1  &-0.2   &1   \\
 0 & -1  & -1  
\end{array}
\right), 
 &  
 A_R =  
 \left(
\begin{array}{ccc}
1  & 0  & 0  \\
0  &1   &0   \\
 0 & 0  & 0 
\end{array}
\right),  
 & 
 C = 
 \left(
\begin{array}{ccc}
0  & -1  & 0  
\end{array}
\right)  
\end{array}
\end{equation}
and a minimum dwell-time $\tau_m$. Since A is a Hurwitz matrix, in fact it has two complex dominant eigenvalues at $-\frac{1}{10} \pm j \frac{3\sqrt{11}}{10}$, stability of the set $\{\mathbf{0}\}\times\{1,-1\}\times [0,\infty)$ will be investigated using Corollary 4.6. After working with SOSTOOLS, a value of $\tau_m^* = 0.6145$ is obtained, and thus global asymptotic stability is guaranteed for any $\tau_m \geq 0.6145$. 

As it is expected, result is somehow conservative. In this case, application of Prop. 4.2  for values $0 < \tau_m \leq 0.6145$ results in that there exist only one fixed point ${\mathbf p} = (-1,0)$ of the mapping $\varPi_g$ (and no other periodic points), and the corresponding matrix $M_{\mathbf p}  = J^\top e^{A\frac{\pi}{\beta}}J$, with $\beta = 3\frac{\sqrt{11}}{10}$, is a Schur matrix. Thus, Prop. 4.2 certifies global assymptotic stability o ${\cal H}_\tau^{cl}$ for smaller values of $\tau_m$ in comparison to Corollary 4.6. Interestingly, for values $\tau_m > \frac{\pi}{\beta} \approx 3.1574$, application of Prop. 4.2 is less direct. For increasing values of $\tau_m$ they appear two fixed points (one sink and a source), periodic-2 points, etc. (details are omitted by brevity). 

\end{example}

\subsection{Reset-times dependent stability conditions: Ranged dwell-time}
If $A$ is not Hurwitz then it is necessary to include a ranged dwell-time $\tau \in [\tau_m,\tau_M]$, enabling to flow only when $\tau \leq \tau_M$. In this case, it is considered  the reset control system $\overline{{\cal H}}_\tau^{cl}$,  which is a modification of ${\cal H}_\tau^{cl}$ including the maximum dwell-time $\tau_M$. $\overline{{\cal H}}_\tau^{cl}$ is given by 
\begin{equation} % \small
\overline{ {\mathcal H}}^{cl}_\tau:
\left \{
\begin{array}{lll}
\dot{\tau} = 1, &
\begin{array}{l}
  \dot{\mathbf{x}}  
\end{array}
=
\begin{array}{l}
  A {\mathbf x} \end{array}  
&\text{, } ({\mathbf x},q,\tau) \in \overline{{\mathcal C}}^{cl}_\tau
\\
\tau^+ = 0, &
\left(
\begin{array}{l}
  {\mathbf x}^+\\
  {q}^+
\end{array}  \right) =
\left(
\begin{array}{cc}
  A_R& 0\\
  0 &-1
\end{array} \right)
\left(
\begin{array}{l}
   {\mathbf x}\\
   {q}
\end{array} \right)&\text{, } ({\mathbf x},q,\tau) \in \overline{{\mathcal D}}^{cl}_\tau
\end{array} 
\right.
\label{eq:HclrhoM}
\end{equation}
where
\begin{equation}
\overline{{\mathcal C}}^{cl}_\tau = \mathcal{C}^{cl}\times [0,\tau_M] \cup \mathcal{D}^{cl}\times [0,\tau_m] 
\label{eq:CclrhoM}
\end{equation}
and 
\begin{equation}
\overline{{\mathcal D}}^{cl}_\tau =  \mathcal{C}^{cl}\times [\tau_M,\infty)  \cup \mathcal{D}^{cl} \times [\tau_m,\infty). 
\label{eq:DclrhoM}
\end{equation}

\vspace{0.5cm}

Now, it is defined the set of time sequences $S_{[\tau_m,\tau_M]}$, given by 
\begin{equation}
S_{[\tau_m,\tau_M]} = \{\{\tau_1, \tau_2, \cdots\} \subset {\mathds R}_{\geq 0}: \tau_i \in [\tau_m,\tau_M] \}
\end{equation}
and it is clear that $S_{\overline{ {\mathcal H}}^{cl}_\tau}  \subset S_{[\tau_m,\tau_M]} $, being $S_{\overline{ {\mathcal H}}^{cl}_\tau}$ defined in a way similar to \eqref{eq:SHtau}. The next Proposition is a direct adaptation of Prop. 4.4  and Corollary 4.5 to $\overline{ {\mathcal H}}^{cl}_\tau$.  

\begin{proposition}
%Consider the reset control system ${\mathcal H}^{cl}_{\tau_m}$ with a Hurwitz matrix $A$. 
The set ${\cal A}_0$ is globally asymptotically stable for 
$\overline{ {\mathcal H}}^{cl}_\tau$,  with $[\tau_m,\tau_M] \subset [\tau_m^*,\tau_M^*]$,  if there exist a sequence of positive definite matrices $\{P_1,P_2, \cdots\}$, such that
\begin{equation}
\begin{array}{l}
  \eta I \leq P_k \leq \rho I    \\
  \\
  e^{A^T \tau_k} A_R P_{k+1} A_R e^{A \tau_k} - P_k \leq -\varepsilon I
\end{array}
\label{eq:prop43aM}
\end{equation}
hold for $k = 1, 2, \cdots$, for some positive constants $\eta$, $\rho$, and $\varepsilon$, and any $\{\tau_1, \tau_2, \cdots\} \in S_{[\tau_m^*,\tau_M^*]} $.
\end{proposition}

Again, if it is considered a sequence of constants matrices $P_k = P >0$ in Prop. 4.8, then it is possible to apply some efficient methods already developed in the literature. For example, in \cite{Banos07,BBbook12} a method for obtaining a set of intervals ${[\tau_m,\tau_M]}$ is developed. Also in \cite{Briat2013} a method based on sum-of-squares conditions is given for this case of ranged dwell-time; the following Corollary is directly based on it. 

 \begin{corollary}
%Consider the reset control system ${\hat {\cal H}}_\tau^{cl}$, given by \eqref{eq:hatHclrho}. The following conditions are equivalent:
%\begin{itemize}
%\item There exist a matrix $P > 0$ such that
%\begin{equation}
%e^{A^T \tau}A_RPA_Re^{A \tau} - P < 0
%\end{equation}
%holds for any $\tau \in [\tau_m,\tau_M]$ .

%\item  There exist a differentiable matrix function $R:[0,\tau_M] \rightarrow \mathds{S}^n$, $R(0) > 0$, and $\varepsilon > 0$ such that
Assume that there exist a differentiable matrix function $R:[0,\tau_M] \rightarrow \mathds{S}^n$, $R(0) > 0$, and $\varepsilon > 0$ such that
\begin{equation}
\begin{array}{l}
A^TR(\theta) + R(\theta)A+ {\dot R}(\theta) \leq 0, \\
A_R R(0) A_R - R(\tau) \leq -\varepsilon I, 
\end{array}
\label{eq:Branged}
\end{equation}
hold for any $\theta \in [0,\tau_M^*]$ and any $\tau \in [\tau_m^*,\tau_M^*]$.
%\end{itemize}
Then the set ${\cal A}_{\mathbf{0}}$ is globally asymptotically stable for $\overline{ {\mathcal H}}^{cl}_\tau$ and for any $[\tau_m,\tau_M] \subset [\tau_m^*,\tau_M^*]$.
\end{corollary}

%\begin{proof}
%
%\end{proof}

\begin{example}
Consider again the reset control system with the Horowitz reset controller (Fig. 6), described in section 4.2. Note that the matrix $A$ is not Hurwitz, since it has an eigenvalue in the closed right half plane. Thus, the reset control system $\overline{ {\mathcal H}}^{cl}_\tau$ is used, forcing jumps when $\tau >\tau_M$. Here,  it is found that \eqref{eq:Branged} is feasible for $\tau_m^* = 0.1$ and $\tau_M^* = 37.5$ ( the sum of squares tool SOSTOOLS, with a polynomial matrix function $R(\theta) = \sum_{i=0} ^6 R_i\theta^i$ of degree 6, has been used). 
Note that the result is somehow conservative: firstly it can be easily shown that stability is also obtained for $\tau_m$ arbitrarly small (see section 4.2); secondly, it is necessary to include a maximum-dwell time $\tau_M < \tau_M^*=37.5$ which is not necessary in the original reset control system ${\cal H}_\tau^{cl}$. In spite of that, note that in this case the reset control system ${\cal H}_\tau^{cl}$ performs jumps with reset intervals upper bounded by $\Delta_M = \frac{2\pi}{\sqrt{19}}+\tau_m$, that for $\tau_m =0.1$, a reasonable value in practice to avoid defective solutions, results in $\Delta_M \approx 1.54$, far below the limit given by $\tau_M^* = 37.5$. In other words, this result guarantees not only that $\overline{ {\mathcal H}}^{cl}_\tau$ is globally asymptotically stable for $[\tau_m,\tau_M] = [0.1,37.5]$, but also that ${\cal H}_\tau^{cl}$ is globally asymptotically stable  for $\tau_m = 0.1$.
\end{example}

\section{Conclusions}
A new model of Clegg integrator,  with an attached error zero-crossing mechanism, has been developed. This results in a reset controller model in the hybrid inclusions framework, enabling to equip the resulting reset control system with good structural properties like robustness against measurement noise and robustness in the stability. The manuscript has been focused on analysis of well-posedness and stability, adapting and extending previous work of the authors to the new reset model. More specifically, stability has been approached by analyzing stability of a Poincar\'e-like map, using two paths: a test based on the eigenvalues of matrices related with periods of reset interval sequences, and quadratic Lyapunov functions-based sufficient conditions. Both approaches have been analyzed in detail, including several examples. Although checking eigenvalues is a simple and efficient test to analyze stability, its applicability depends on the computation of a finite number of periodic points; as an interesting result, it has been formally shown that this is not always possible since in some cases reset intervals may produce chaotic sequences. As an alternative, Lyapunov function-based results may be applied: different results has been obtained for the case in which the base control system is stable or unstable. As a final conclusion, it is believed that the manuscript gives a solid framework for reset control systems with a zero-crossing resetting law, that may serve as a basis for new theoretical and practical advances.

\section*{Acknowledgments}   
It is gratefully acknowledged the helpful comments of Andrew R. Teel on an early version of the reset controller model developed in this work.  

\bibliographystyle{elsarticle-harv}
\bibliography{biblio}
%\bibliography{biblio}

\end{document}